\documentclass[letterpaper, twocolumns, 10 pt]{IEEEtran}  
\usepackage[left=0.7in,right=0.7in,top=1in,bottom=1in]{geometry}        
\usepackage[utf8]{inputenc}
\usepackage{subcaption}
\usepackage{graphicx}
\usepackage{amsmath}
\usepackage{amsthm}

\usepackage{amssymb}
\usepackage[version=4]{mhchem}
\usepackage{siunitx}
\usepackage{longtable,tabularx}
\usepackage{xcolor}
\usepackage{cite}
\usepackage{authblk}
\usepackage{amsfonts,color}

\newtheorem{theorem}{Theorem}
\newtheorem{corollary}{Corollary}
\newtheorem{lemma}{Lemma}
\newtheorem{remark}{Remark}
\newtheorem{definition}{Definition}
\newtheorem{assumption}{Assumption}

\ifCLASSINFOpdf

\else

\fi

\usepackage{duckuments}

\begin{document}

\title{\huge Mutual Support by Sensor-Attacker Team for a Passive Target}

 \author{Prajakta Surve$^{1}$,~ Shaunak D.~Bopardikar$^{1}$, ~  Alexander Von Moll$^{2}$, \\ \vspace{-1em} Isaac Weintraub$^{2}$, and ~David W.~Casbeer$^{2}$
\thanks{*This research was supported by the Aerospace Systems Technology Research and Assessment (ASTRA) Aerospace Technology Development and Testing (ATDT) program at AFRL under contract number FA8650-21-D-2602.
DISTRIBUTION STATEMENT  A. Distribution is unlimited. AFRL-2025-2215, Cleared 01 May 2025.}
\thanks{$^{1}$Prajakta Surve and Shaunak D. Bopardikar are with Michigan State University, East Lansing, MI 48823 USA (e-mails: {\tt\small survepra@msu.edu, shaunak@egr.msu.edu})}%
\thanks{$^{2}$Alexander Von Moll, Issac Weintraub and David W. Casbeer are with Control Science Center, Air Force Research Laboratory, Ohio, USA (e-mail: {\tt\small alexander.von\_moll@afrl.af.mil, isaac.weintraub.1@afrl.af.mil, david.casbeer@us.af.mil}).}%
}
\maketitle


\begin{abstract}
We introduce a pursuit game played between a team of a sensor and an attacker and a mobile target in the unbounded Euclidean plane. The target is faster than the sensor, but slower than the attacker. The sensor's objective is to keep the target within a sensing radius so that the attacker can capture the target, whereas the target seeks to escape by reaching beyond the sensing radius from the sensor without getting captured by the attacker.  We assume that as long as the target is within the sensing radius from the sensor, the sensor-attacker team is able to measure the target's instantaneous position and velocity. We pose and solve this problem as a \emph{game of kind} in which the target uses an open-loop strategy (passive target).  
Aside from the novel formulation, our contributions are four-fold. First, we present optimal strategies for both the sensor and the attacker, according to their respective objectives. Specifically, we design a sensor strategy that maximizes the duration for which the target remains within its sensing range, while the attacker uses proportional navigation to capture the target. Second, we characterize the \emph{sensable region} -- the region in the plane in which the target remains within the sensing radius of the sensor during the game -- and show that capture is guaranteed {if and only if} the Apollonius circle between the attacker and the target is fully contained within this region. Third, we {derive a lower bound} on the target's speed below which capture is guaranteed, and an upper bound on the target speed above which there exists an escape strategy for the target, from an arbitrary initial orientation between the agents. Fourth, for a given initial orientation between the agents, we present a sharper upper bound on the target speed above which there exists an escape strategy for the target.
\end{abstract}


\section{Introduction}\label{SecI}
Consider a scenario where a surveillance UAV, operating with a limited sensing radius, is deployed to track a target -- an adversarial vehicle trying to escape a protected area. Simultaneously, {an interceptor vehicle is tasked with reaching the mobile target. However, this vehicle relies on the UAV to provide it with the target's location and velocity at all times. So how should the UAV move so that it can keep the target within its sensing range for as long as possible to enable a successful interception? Or can the target outrun the surveillance UAV and escape interception?} This interplay between sensing constraints and vehicle kinematics in this scenario of \emph{mutual support} poses significant strategic challenges.
This paper formalizes this interaction as a differential pursuit-evasion game between the three agents who have different capabilities and objectives.

\subsection{Related Work}
Works such as \cite{zha2016construction,chen2016multi,garcia2021cooperative,ramana2017pursuit} examined scenarios involving multiple attackers against a single, faster evader.   
For sustained target-tracking, authors in \cite{weintraub2023surveillance} addressed the problem of maximizing observation time of a faster, fixed-course target by a slower observer. The study derived closed-form strategies for both approach and surveillance phases by using  a two-phase optimal control framework.
Similarly, in \cite{awheda2016decentralized}, authors utilized Apollonius circles to define capture regions in a multi-attacker, single-superior-evader setting. Authors in \cite{makkapati2019optimal,kumar2025cooperative} also used a geometric approach by constructing Voronoi-based partitions for assigning attackers. Work done in \cite{yan2022matching} presented matching-based strategies for coordinating multiple heterogeneous pursuers in 3D reach-avoid games, enabling effective defense of a target region through subgame decomposition and dynamic assignment. Additionally, in \cite{garcia2017geometric}, a maritime pursuit problem involving two cutters and a fugitive ship was investigated, which was later generalized in \cite{von2018pursuit} to accommodate coordinated pursuit in more dynamic environments. Recently, in \cite{DOROTHY2024111587} authors have proved that in many pursuit-evasion scenarios, the evader cannot escape beyond the initial Apollonius Circle, enabling simple pursuer strategies that guarantee capture within this region.

Complementing these lines of research, the three-agent Target-Attacker-Defender (TAD) games have also garnered significant attention over the past few years. In these differential games, the defender aims to capture the attacker before it reaches the target. Cooperative optimal control strategies for this setup were explored in \cite{4101686,4102335}. Similarly in \cite{9873915} the authors have addressed differential games with one or two defenders constrained to a circular boundary, deriving analytical strategies and conditions for successful interception. The Triangle Guidance strategy, introduced in \cite{4634cdcbd9b240a38024275f60ffbe84,GARCIA201714200}, directs the defender to stay aligned along the line of sight between the attacker and the target, with the target following a predefined trajectory. {The work in \cite{9872043} discussed a TAD game involving multiple defenders and analyzed two variations based on asymmetric information among the three agents. In both scenarios, the attacker is assumed to have unlimited visibility, while the target and defenders operate under visibility constraints.} In \cite{das2024heterogeneous} authors demonstrated that in a two-versus-two target defense game, attackers employing heterogeneous roles -- where one attacker assists by intercepting a defender -- can outperform traditional assignment-based strategies that treat engagements independently. Additionally, non-cooperative TAD scenarios -- particularly those analyzing the attacker-target miss distance when the target and defender do not cooperate -- have been investigated in works such as \cite{article}. Most of the available TAD literature assumes that the defender has capabilities equal to or superior to the attacker, as well as complete state awareness, conditions that do not hold in mutual support scenarios.

\subsection{Contributions}
This paper presents a three-agent pursuit-evasion problem involving a sensor, an attacker and a target, {in which the sensor has a limited sensing radius within which it can detect the target's location and speed}. {Our assumption that the attacker relies entirely on real-time sensing reflects how certain practical systems operate — for instance, in beam-rider guidance \cite{anjaly2016beam}, where the attacker follows a laser or radar beam steered by an external sensor. These systems lack onboard tracking and are highly vulnerable to disruptions like occlusion, jamming, or sensor failure. Our model captures this dependence, treating the attacker as a sensor-driven agent. } Beyond sensing limitations, the agents differ in speed, with the sensor being slower than both the target and the attacker.  
Although works such as \cite{bopardikar2008discrete,lin2015nash,9966171,maity2024optimal} do take into account the attacker's limited observation capabilities, the attackers in these works either have a speed advantage over target or do not depend on a slower sensing agent. More recently, authors in \cite{HUANG2025112258} employed geometric methods based on Apollonius circles to study dominance regions under non-anticipatory information patterns in a two-agent pursuit-evasion game.  In contrast, we present a more complex three-agent problem where the attacker depends on a cooperating sensor that is slower and has a restricted sensing range, {thereby adding an additional layer of complexity in the information access.}

{
The analysis in this paper focuses on a 2D environment with vehicles that can directly control their heading (that is, they are not constrained by turn rates). 
This setup reflects beyond-visual-range scenarios \cite{garcia2021beyond}, where the horizontal distances between vehicles are much larger than their altitude differences or turn radii.
Thus, the latter are assumed to be negligible.
}

In this paper, the attacker’s goal is to intercept the target, while the sensor’s role is to keep the target within its sensing radius for as long as possible. {The target wins the game if it gets outside the sensor's sensing radius without getting intercepted by the attacker.} 
{The key contributions of this work are summarized below
\begin{enumerate}
\item {\bf A novel hierarchical three-agent \emph{Game of Kind} formulation:} We introduce a new pursuit-evasion game involving three agents with distinct roles: a target, a sensor, and an attacker. The sensor is responsible for detecting the target, while the attacker can only engage once detection occurs, establishing a clear sensor-attacker hierarchy. We consider open-loop strategies for the target. In such strategies, the target selects its heading at the beginning of the game. This heading may be chosen in an adversarial manner by the target, but once chosen, the heading remains constant throughout the engagement. The target selects a constant open-loop heading at the start of the game, which leads to a binary outcome game — either capture (attacker intercepts after detection) or escape (target avoids detection entirely).
\item {\bf Geometry-driven sensor and attacker strategies:} We derive an optimal strategy for the sensor that maximizes the time during which the target remains within the sensing region. We define the \emph{sensable region} — a geometric construct that captures the set of reachable target positions within the sensor’s sensing range over time. The attacker strategy is designed using the Apollonius circle geometry to achieve minimum-time interception. 
\item  {\bf Geometric and speed-based conditions for capture and escape:} We provide analytic expressions for escape time and derive critical thresholds on the target’s speed based on the initial configuration of the agents. Specifically, 
\begin{enumerate}
    \item We show that capture occurs if and only if the Apollonius circle between the attacker and the target lies entirely within the sensable region at the start.
    \item We derive an upper bound on the target’s speed that guarantees capture, regardless of the initial orientation between the agents.
    \item We derive a lower bound on the target’s speed above which escape is always possible, regardless of the initial orientation between the agents.
    \item For a fixed initial orientation, we derive a sharper lower bound that guarantees escape based on the initial orientation between the players. 
\end{enumerate}
Together, these results provide a comprehensive set of novel thresholds that characterize the outcome of the engagement.
\end{enumerate}}

\noindent {\bf Organization of this paper:} We formally define the problem setup and the key assumptions for our analysis in Section \ref{SecII}. We discuss the detailed development of the agent strategies and their properties in Section \ref{SecIII}. We present the main results in Section \ref{SecIV}. We demonstrate numerical illustrations to visualize the proposed results in Section \ref{SecV}. Finally, we conclude with future directions in Section \ref{SecVI}.


\section{Problem Formulation}\label{SecII}

We consider a cooperative pursuit-evasion scenario involving a Sensor ($S$), an Attacker ($A$), and a Target ($T$). The engagement begins with the following initial condition: 1) the sensor detects the target if the target {is within a sensing radius} $R$ from the sensor, and 2) the attacker attempts to intercept the target while the target is within the sensing radius. {The target seeks to end the engagement by increasing its distance from the sensor to a value more than $R$ while avoiding the attacker.}

\medskip
\noindent {\bf Notation:} Throughout this paper, we refer to the sensor, attacker, and target as $S$, $A$, and $T$, respectively. For the engagement parameters, we adopt the convention of using superscripts to indicate physical entities or variables (such as sensor, attacker, target), and subscripts to denote time. If a variable does not have a time subscript, then it is assumed to be constant. For example,  
$d_t^{ST}$ represents {the Euclidean distance} between the sensor $S$ and the target $T$ at time $t$. Certain engagement parameters such as the heading or the speed may include an additional superscript ``$\star$'' to signify an optimal or critical value. For instance, $\gamma^{S\star}$ denotes a constant optimal heading for the sensor. 

\medskip
The positions of all three agents $S$, $A$ and $T$  are expressed in {a fixed global} coordinate frame, illustrated in Fig.~\ref{fig:eng}. The initial position of $S$ is set at the origin of this reference frame. At any time $t \in [0,\, t_f]$, where $t_f$ denotes the termination time of the engagement, the line-of-sight (LOS) angles from $A$ to $T$ and from $S$ to $T$ are denoted by $\theta^{AT}_t$ and $\theta^{ST}_t$, respectively. The corresponding relative distances at time $t$ are denoted by $d^{AT}_t$ and $d^{ST}_t$.

Each agent $i \in \{S, A, T\}$ moves with a constant speed $v^i$ and follows a fixed heading angle $\gamma^i$. The target chooses its heading ($\gamma^T$) possibly in an adversarial manner at the beginning of the engagement, that is, at $t = 0$, and maintains this constant heading throughout, regardless of the strategies employed by $S$ or $A$. We refer to such targets as \emph{passive targets}. This assumption is standard in the pursuit-evasion literature \cite{boyell1976defending,  weintraub2021maximum, weintraub2021engagement, vonmoll2022pure}, where it enables analytical tractability and reflects scenarios in which a non-maneuverable target selects its heading at the start of the engagement and follows an open-loop trajectory, uninfluenced by the attacker's actions. {The target is modeled as \emph{passive} as it is unable to perform deceptive or complex maneuvers. This assumption helps us focus on the core pursuit-evasion kinematics and obtain clear analytical results.}


Let $i_t = (i^x_t, i^y_t)$ denote the $x$ and $y$ coordinates of agent $i$. The agents’ kinematics are governed by the following equations.
\begin{align}\label{dyn}
\begin{split}
     \dot{S}^{x}_t &=   v^S \cos \gamma^S, ~~
    \dot{S}^{y}_t =   v^S \sin \gamma^S,\\
    \dot{A}^{x}_t &=     v^A \cos \gamma^A,~~
    \dot{A}^{y}_t = v^A \sin \gamma^A,\\
    \dot{T}^{x}_t &= v^T \cos \gamma^T,~~ 
    \dot{T}^{y}_t =v^T \sin \gamma^T.
\end{split}
\end{align}
Note that $\gamma^i$ also denotes the {control input} for the corresponding agent $i$.

\begin{figure}[!h]
    \centering
\includegraphics[width=0.75\linewidth]{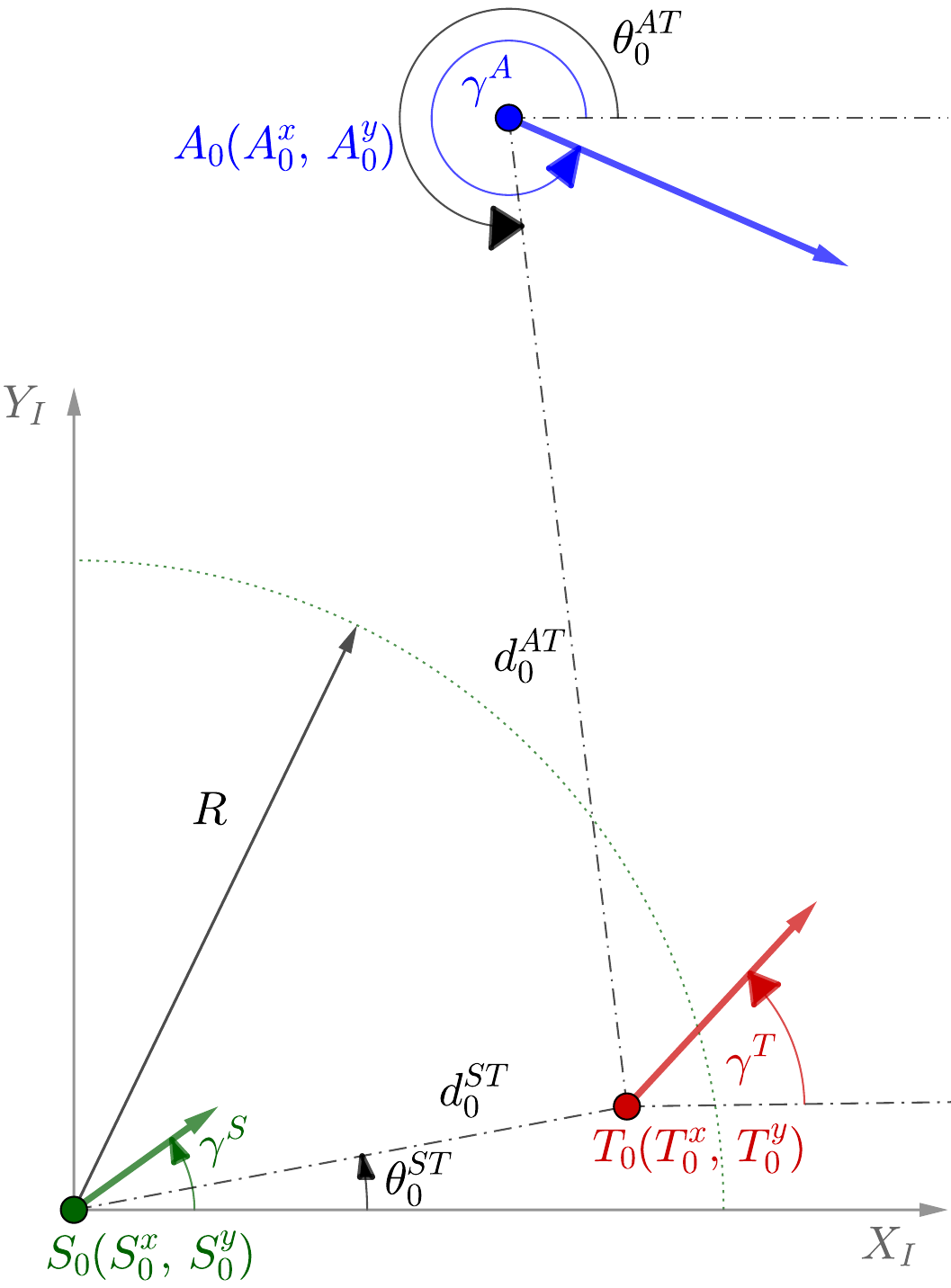}
    \caption{Illustration of the Sensor-Attacker-Target (SAT) engagement scenario, showing the relative positions of the agents.}
    \label{fig:eng}
\end{figure}

\medskip
The engagement terminates upon satisfying one of the following conditions:
\begin{itemize}
    \item \textbf{Escape}: {The distance between the target $T$ and the sensor $S$ equals or exceeds $R$ while the attacker $A$ and the target locations have not been coincident. Mathematically, this means there exists a time instant ${t}_f$ such that
    \[
    d^{ST}_{t_f} \geq R, \quad \text{ and } \quad d^{AT}_t > 0,~ \forall\, t \in [0,t_f].
    \]}
    
    \item \textbf{Capture}: {The distance between the target $T$ and the sensor $S$ remains less than $R$ while the attacker $A$ and the target $T$ locations become coincident. Mathematically, this means that there exists a time instant $t_f$ such that
    \[
    d^{AT}_{t_f} = 0, \quad \text{ and } \quad d^{ST}_t < R,~ \forall\, t \in [0,t_f].
    \]
    }
\end{itemize}

\medskip
{Next, we introduce assumptions on the information pattern and the speeds of the agents in this formulation.}

\begin{assumption}[Information Pattern]\label{as:info}
   The {sensor-attacker team} has access to the target’s current location $T_t\,(T^x_t, T^y_t)$, speed ($v^T$) and the heading ($\gamma^T$) as long as the target remains within the sensor’s sensing range, defined by a disc of radius $R$. 
   {The target, on the other hand, has full knowledge of the initial positions of both the sensor and the attacker at time $t=0$, as well as the strategies employed by $S$ and $A$.} \hfill $\blacksquare $
\end{assumption}

As mentioned previously, the targets in consideration in this study are \emph{passive targets} in that they do not adapt their trajectory in response to the attacker or sensor, even if they have access to information about their strategies.

\begin{assumption}[Agent Speeds] \label{as:speeds}
   The agents' speeds satisfy
   \[
   v^S < v^T < v^A,
   \]
   with the {attacker} being the fastest and the {sensor being} the slowest among the agents. \hfill $\blacksquare $
\end{assumption}

Note that Assumption \ref{as:speeds} ensures the problem remains non-trivial. For example, if $T$ were slower than both $A$ and $S$, then capture would be inevitable. Conversely, if $T$ were faster than both $A$ and $S$, then it would always be able to escape, rendering the engagement trivial. Assumption~\ref{as:speeds} makes the problem meaningful and well-posed. 

Without loss of generality, the attacker's speed is normalized to $v^A = 1$. The Sensor-to-Target speed ratio is defined as $\nu = \dfrac{v^S}{v^T} < 1$, and the Target-to-Attacker speed ratio as $\mu = \dfrac{v^T}{v^A} < 1 $.

\medskip
We describe the problem of determining whether the target $T$ is captured (that is, the team of $A$ and $S$ wins) or successfully evades capture (that is, $T$ wins), as a \emph{game of kind}. We will begin by deriving a strategy for $S$ to ensure that $T$ remains within the sensing radius as long as possible. As long as $T$ is within the sensing distance from $S$, {$A$ will follow a strategy to capture $T$ in minimum time as derived later in the manuscript.} Our objective then is to characterize the conditions on the problem parameters that guarantee either capture or successful evasion.

\section{Agent strategies and their properties}\label{SecIII}
In this section, we {present candidate strategies that the sensor and the attacker should follow, based on their respective objectives. We also present their respective analytic properties.}
\subsection{Strategy for $S$}
As depicted in Fig. \ref{fig:dirS}, the sensor $S$ engages with the target $T$, which moves at a constant heading $\gamma^T$ known to $S$. For capture, the sensor must maintain the target within its sensing radius $R$.

\medskip
\begin{lemma}[\textbf{Strategy for $S$}]\label{lemma1}
     Let the initial positions of the sensor and target be $S_0 = (S^{x}_0, S^{y}_0)$ and $T_0 = (T^{x}_0, T^{y}_0)$, respectively. Suppose that the target moves with a constant heading $\gamma^T$, known to the sensor, and that the engagement ends at time $t = t_f$, with the respective final positions of $S$ and $T$ denoted as $S_{t_f} = (S^{x}_{t_f}, S^{y}_{t_f})$ and $T_{t_f} = (T^{x}_{t_f}, T^{y}_{t_f})$. Then,
     \begin{enumerate}
         \item The points $S_0, S_{t_f}$ and $T_{t_f}$ are collinear and in that order with $d^{ST}_{t_f} = \|S_{t_f}-T_{t_f}\| = R$, and
         \item the optimal sensor heading $\gamma^{S\star}$ that maximizes the duration for which the target remains within the sensing radius is given by  
\begin{align}\label{eq:gammastar}
    \gamma^{S\star} = \text{atan2} \left( T_{t_f}^y - S_{0}^y, \, {T_{t_f}^x - S_{0}^x} \right), 
\end{align}
where $\text{atan2}(\cdot,\cdot)$ is the four quadrant inverse tangent function.

     \end{enumerate}
\end{lemma}

\begin{proof}{Proof of Lemma \ref{lemma1}}
  \begin{figure}[!h]
        \centering
        \includegraphics[width=0.95\linewidth]{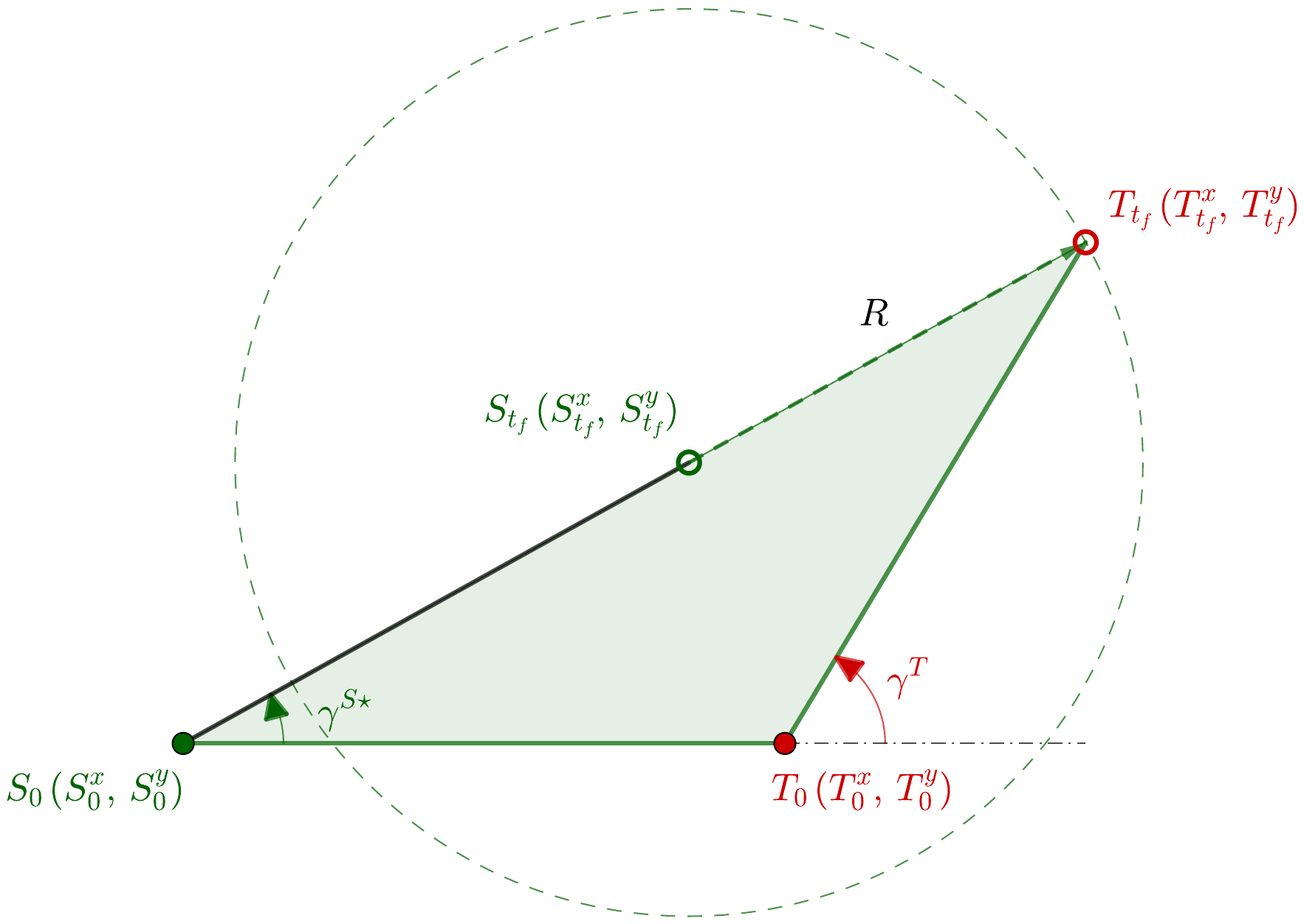}
        \caption{Optimal direction for $S$.}
        \label{fig:dirS}
    \end{figure}

The state kinematics of the agents $S$ and $T$ are given by Eq. (\ref{dyn}).
Let us denote the relative position of $S$ with respect to $T$ by $(X_t,\,Y_t)$ at any given time $t$ as
\begin{align}\label{relX}
    X_t &= T_t^x - S_t^x,\\ \label{relY}
    Y_t &= T_t^y - S_t^y.
\end{align}
To maximize the time of the engagement, we define the cost function
\begin{align}\label{j}
    J = \int_0^{t_f} \mathrm{d}t,
\end{align}
where, as mentioned earlier, $t_f$ denotes the final time of the engagement where either capture or escape occurs, and there is no terminal cost.

{Then,} the Hamiltonian for this problem is given by
\begin{align}\label{H}
    \mathcal{H} = 1 + \lambda^x_t \dot{X}_t + \lambda^y_t \dot{Y}_t,
\end{align}
where $\lambda^x_t \in \mathbb{R}$ and $\lambda^y_t \in \mathbb{R}$ are the costate variables at time $t$. The evolution of the costate variables satisfies 
\begin{align}\label{ldtxy}
    \dot{\lambda}^x_t &= -\frac{\partial \mathcal{H}}{\partial X_t}, ~~\text{and}~~ \dot{\lambda}^y_t = -\frac{\partial \mathcal{H}}{\partial Y_t}.
\end{align}
Since $\mathcal{H}$ does not explicitly depend on $X_t$ or $Y_t$, it follows that
\begin{align}\label{constlxy}
    \dot{\lambda}^x_t = \dot{\lambda}^y_t = 0.
\end{align}
Thus, $\lambda^x$ and $\lambda^y$ remain constant throughout the motion.
By Pontryagin's Maximum Principle \cite{pontryagin2018mathematical}, the optimal heading $\gamma^{S\star}$ for $S$ that maximizes the Hamiltonian $\mathcal{H}$, satisfies
\begin{align}
    \frac{\partial \mathcal{H}}{\partial \gamma^{S\star}} = 0.
\end{align}
Taking the partial derivative of $\mathcal{H}$ with respect to $\gamma^{S\star}$ and setting it to zero yields:
\begin{align}
     \lambda^x (v^S \sin \gamma^{S\star}) + \lambda^y (- v^S \cos \gamma^{S\star}) &= 0,\\ \label{tanGm}
  \implies    \tan \gamma^{S\star} &= \frac{\lambda^y}{\lambda^x}.
\end{align}
Since $\lambda^x$ and $\lambda^y$ are constants and their values depend on the final positions of $T$ and $S$, the sensor heading $\gamma^{S\star}$ is also a constant throughout the engagement. 

To determine $\lambda^x$ and $\lambda^y$, we apply the transversality condition from optimal control theory. {Define} $\phi(S^x_{t_f}, S^y_{t_f}, T^x_{t_f}, T^y_{t_f})$ as the terminal manifold that defines the end of the engagement as 
\begin{align}
    \phi = (S^{x}_{t_f} - T_{t_f}^{x})^2 + (S_{t_f}^y - T_{t_f}^y)^2 - R.
\end{align}
{Then, at the end of the engagement, the target is on the boundary of the sensing disc at $t = t_f$ is equivalent to $\phi = 0$.} The transversality condition {in the case of no terminal cost} is given by
\begin{align}
    \lambda_{t_f} = [\lambda^x_{t_f},\, \lambda^y_{t_f}] = \eta \frac{\partial \phi}{\partial \mathbf{p}_{t_f}},
\end{align}
where $\eta \in \mathbb{R}$ is a constant and 
\begin{align}
    \mathbf{p}_{t_f} = [T_{t_f}^x - S_{t_f}^x, T_{t_f}^y - S_{t_f}^y]^\top.
\end{align}
Therefore, at time $t_f$,
\begin{align} \label{lx}
    \lambda^x_{t_f} &= 2\eta (S_{t_f}^{x} - T_{t_f}^x), \\ \label{ly}
    \lambda^y_{t_f} &= 2\eta (S_{t_f}^y - T_{t_f}^y).
\end{align}
From Eq. (\ref{tanGm}), and the fact that $\gamma^{S\star}$ is a constant throughout the engagement, the optimal heading direction is given by
\begin{align} \label{cGm}
    \cos \gamma^{S\star} &= \frac{\lambda^x_{t_f}}{\sqrt{{\lambda^x_{t_f}}^2 + {\lambda^y_{t_f}}^2}}, \\ \label{sGm}
    \sin \gamma^{S\star} &= \frac{\lambda^y_{t_f}}{\sqrt{{\lambda^x_{t_f}}^2 + {\lambda^y_{t_f}}^2}}. 
\end{align}
Substituting Eqs. (\ref{lx}) and (\ref{ly}) into Eqs. (\ref{cGm}) and (\ref{sGm}), we obtain:
\begin{align}
    \cos \gamma^{S\star} &= \frac{2 \eta \left(S_{t_f}^x - T_{t_f}^x\right)}{\sqrt{\left[2 \eta (S_{t_f}^x - T_{t_f}^x)\right]^2 + \left[2 \eta (S_{t_f}^y - T_{t_f}^y)\right]^2}}, \\
    \sin \gamma^{S\star} &=  \frac{2 \eta \left(S_{t_f}^y - T_{t_f}^y\right)}{\sqrt{\left[2 \eta (S_{t_f}^x - T_{t_f}^x)\right]^2 + \left[2 \eta (S_{t_f}^y - T_{t_f}^y)\right]^2}}.
\end{align}
The above equations can be further simplified as
\begin{align} \label{cGmstar}
    \cos \gamma^{S\star}  &= -\text{sign} (\eta) \left(\dfrac{T_{t_f}^x - S_{t_f}^x}{R}\right)\\ \label{sGmstar}
   \sin \gamma^{S\star}  &=  -\text{sign} (\eta)  \left(\dfrac{T_{t_f}^y - S_{t_f}^y}{R} \right).
\end{align}
Since $\gamma^{S\star}$ is the heading of $S$ pointing at $T$, we assume $\text{sign}(\eta) = -1$, which further leads to
\begin{align}
     \cos \gamma^{S\star}  &=  \dfrac{T_{t_f}^x - S_{t_f}^x}{R}, ~\text{and}~
    \sin \gamma^{S\star}  =  \dfrac{T_{t_f}^y - S_{t_f}^y}{R}.
\end{align}
Since $\gamma^{S\star}$ is constant, $S$ follows a straight line trajectory, which also implies that $S_0$, $S_{t_f}$ and $T_{t_f}$ must be collinear as depicted in Fig. \ref{fig:dirS}. Thus, the optimal heading angle $\gamma^{S\star}$ is given by:
\begin{align}
    \gamma^{S\star} = \text{atan2} (T_{t_f}^y - {S_{0}^y}, \, T_{t_f}^x - {S_{0}^x}).
\end{align}
This completes the proof.
\end{proof}
\medskip 
Recall that for $T$ to be captured, {its distance from $S$ has to be less than} $R$. Otherwise, it is considered to have successfully escaped, ending the engagement. The following result provides an expression for the time $t_f$ required for $d^{ST}_{t_f} = R$.

\begin{theorem}[\textbf{Time to Escape}]\label{thm1}
{From given initial locations of $S$ and $T$}, if the sensor $S$ follows the optimal heading $\gamma^{S\star}$ as given in 
Eq.~\eqref{eq:gammastar} while $T$ moves with a constant heading $\gamma^T$ and speed $v^T$, {then time required by the target to escape the detection and thereby ending the engagement, is given by}  
\begin{align}\label{tesc}
\begin{split}
    t_f &= \dfrac{ \Omega + \sqrt{\Omega^2 - (1 - \nu^2)({d^{ST}_0}^2-R^2)}}{v^T(1-\nu^2)} ,
\end{split}
\end{align}
where $\Omega = \nu R - d^{ST}_0 \cos(\gamma^T - \theta^{ST}_0)$. \flushright $\blacksquare$
\end{theorem}
\begin{proof}[Proof of Theorem \ref{thm1}]
 Consider the engagement between $S$ and  $T$ at $t = t_f$, as illustrated in Fig.~\ref{fig:esct}. 
 If $\|S_0 S_{t_f}\| = \|S_0 - S_{t_f}\|$ and $\|T_0 T_{t_f}\| = \|T_0 - T_{t_f}\|$ represent the distances traveled by $S$ and $T$ during the interval $[0,\,t_f]$, respectively, then equating their travel times gives
\begin{align}\label{nu}
    \dfrac{\|S_0 S_{t_f}\|}{v^S} =  \dfrac{\|T_0 T_{t_f}\|}{v^T} \implies \|S_0S_{t_f}\| = \nu \|T_0T_{t_f}\|,
\end{align}
where $\nu = \dfrac{v^S}{v^T}$ is the speed ratio of $S$ and $T$.

Assuming that the target reaches the boundary of the sensing disc at $T_{t_f}$, as shown in Fig. \ref{fig:esct}, we apply the law of cosines to $\triangle S_0\, T_{t_f}\, T_0$, yielding
\begin{align}\label{esc1}
\nonumber
        (\|S_0 S_{t_f}\| + R)^2 &= {d^{ST}_0}^2+ \|T_0 T_{t_f}\|^2 \\ & \hspace{0.15cm}- 2\,d^{ST}_0 \,\|T_0 T_{t_f}\| \cos \left[\pi - (\gamma^T - \theta^{ST}_0) \right],
\end{align}
\begin{figure}[!h]
    \centering
    \includegraphics[width=0.85\linewidth]{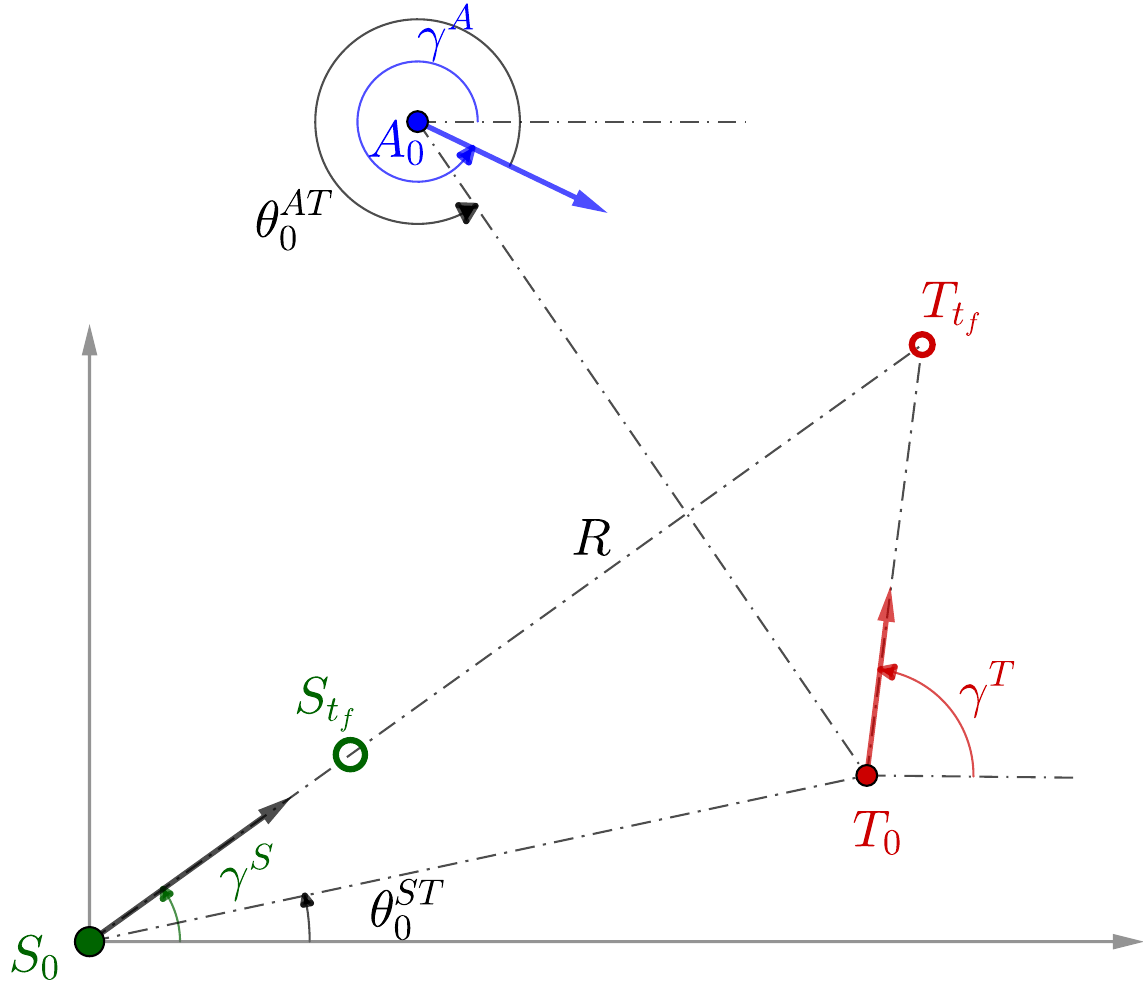}
    \caption{Engagement among the three agents at the time of escape.}
    \label{fig:esct}
\end{figure}
Substituting the value of $\|S_0 S_{t_f}\|$ from Eq.~\eqref{nu} and simplifying, we obtain
\begin{align}\label{esc2}
\nonumber
        (\nu \|T_0 T_{t_f}\| + R)^2 &= {d^{ST}_0}^2 + \|T_0 T_{t_f}\|^2 \\ &\hspace{0.2cm} + 2\,d^{ST}_0 \,\|T_0 T_{t_f}\| \cos \left(\gamma^T - \theta^{ST}_0 \right)
    \end{align}
Rearranging this results in the quadratic form of $\|T_0 T_{t_f}\|$ gives
\begin{align}\label{quad}
  \nonumber 0 &= (1-\nu^2) \|T_0 T_{t_f}\|^2  \\ &\hspace{0.1cm} + 2 (d^{ST}_0 \cos (\gamma^T - \theta^{ST}_0) - \nu R) \|T_0 T_{t_f}\| + {d^{ST}_0}^2 - R^2. 
\end{align}
Solving this quadratic equation gives the general solution for $\|T_0 T_{t_f}\|$ as
   \begin{align}\label{ttf-roots}
    \begin{split}
     & \| T_0 T_{t_f}\| = \dfrac{\nu R - d^{ST}_0 \cos(\gamma^T - \theta^{ST}_0) }{(1-\nu^2)}  \\& \hspace{0.15cm}   \pm \dfrac{\sqrt{(\nu R - d^{ST}_0 \cos(\gamma^T - \theta^{ST}_0))^2 - (1 - \nu^2)({d^{ST}_0}^2-R^2)}}{(1-\nu^2)}.
   \end{split}
   \end{align}
Since $\nu < 1$ and ${d^{ST}_0}^2 - R^2 < 0$, the discriminant is always non-negative, ensuring real solutions. Retaining only the positive root for physical feasibility, we have

    \begin{align}\label{ttf}
    \begin{split}
     & \| T_0 T_{t_f}\| = \dfrac{\nu R - {d^{ST}_0} \cos(\gamma^T - \theta^{ST}_0) }{(1-\nu^2)}  \\& \hspace{0.15cm} + \dfrac{\sqrt{(\nu R - d^{ST}_0 \cos(\gamma^T - \theta^{ST}_0))^2 - (1 - \nu^2)({d^{ST}_0}^2-R^2)}}{(1-\nu^2)}.
   \end{split}
   \end{align}
Finally, using $\|T_0 T_{t_f}\| = v^T t_f$ and $\Omega  = \nu R - {d^{ST}_0} \cos(\gamma^T - \theta^{ST}_0) $ we obtain Equation~\eqref{tesc}. 
This completes the proof.
   \end{proof}
\medskip   
Equation~\eqref{tesc} provides the time required for the target to escape $R$, given the engagement geometry and relative speeds. At this point, we would like to define the term \emph{sensable region}.
\begin{definition}[Sensable Region]\label{def:sensable}
{The \emph{sensable region} (denoted by $\mathcal{S}_t$) is the set of all points that the target can reach by moving at a constant speed $v^T$ and for some constant heading $\gamma^T \in [0,\, 2\pi)$ over the time interval $[0, t_f(\gamma^T)]$. Here, $t_f(\gamma^T)$ is the escape time given in Eq. (\ref{tesc}), marking the point at which the engagement ends with the target's escape. }

\hfill ${\blacksquare}$
\end{definition}

Figure \ref{fig:sregion} illustrates the Cartesian oval shape of the sensable region  $\mathcal{S}_0$  for all values of the heading $\gamma^T \in [0, \, 2\pi)$. 
\begin{figure}[!h]
    \centering
\includegraphics[width=\linewidth]{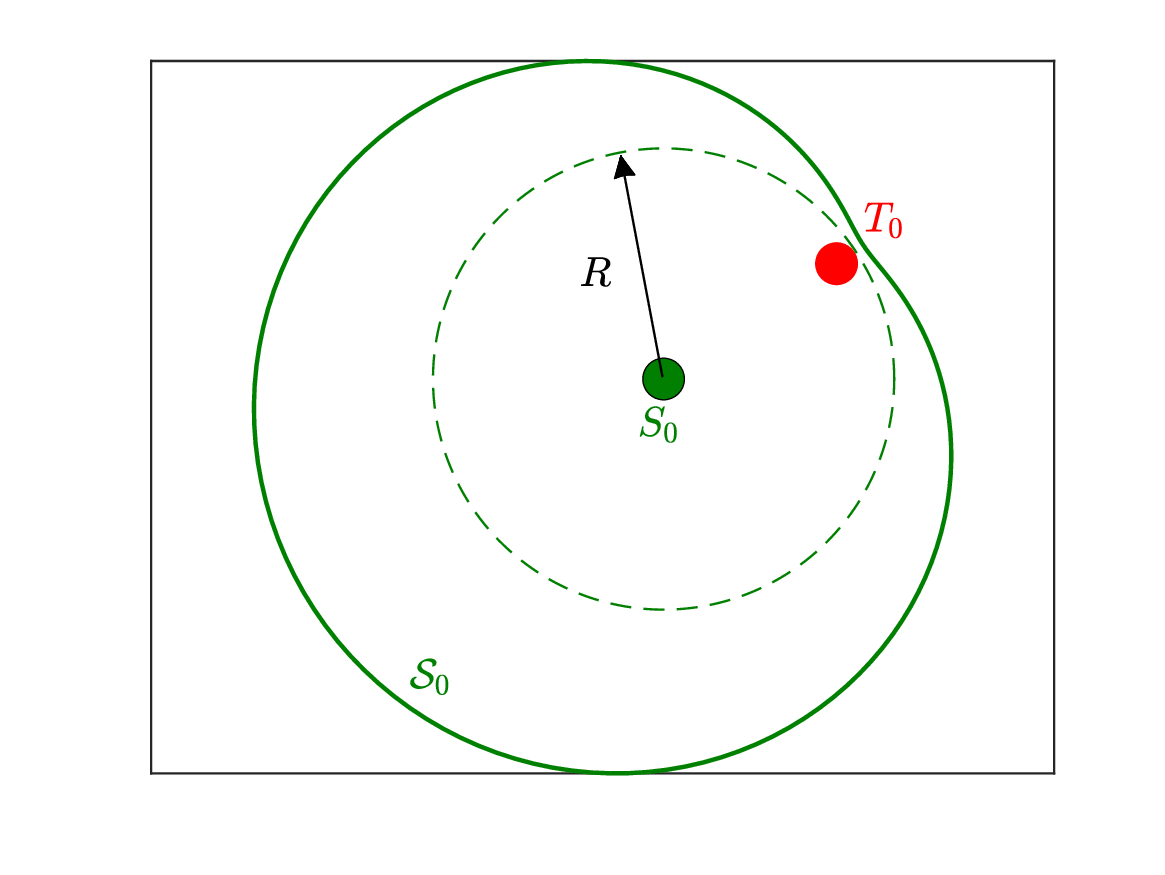}
    \caption{Visualization of the sensable region  $\mathcal{S}_0$  for the given initial positions of the agents and all values of the target heading $ \gamma^T \in [0, 2\pi) $. }
    \label{fig:sregion}
\end{figure}
\medskip
Theorem \ref{thm1} yields an expression for the minimum time the target $T$ needs to escape the sensable region, as summarized in the following corollary.
\begin{corollary}[\textbf{Minimum Time to Escape}]\label{cor1}
    The minimum time required for the target to escape the sensable region $\mathcal{S}_0$ from $t = 0$,  given the initial conditions, is achieved when the target moves directly away from the sensor, that is, $\gamma^T = \theta^{ST}_0$. In this case, the minimum escape time is given by 
    \begin{align}\label{tfmin}
        t_f^{\min} = \dfrac{R - d^{ST}_0}{v^T (1 - \nu)}.
    \end{align}
\end{corollary}
\begin{proof}[Proof of Corollary \ref{cor1}]
  To determine the target heading $\gamma^T$ that minimizes $\|T_0 T_{t_f}\|$, we differentiate Eq.~\eqref{ttf} with respect to $\gamma^T$ and set the derivative equal to zero. This is equivalent to differentiating the quadratic expression in Eq.~\eqref{quad} with respect to $\gamma^T$ and setting it equal to zero. Therefore,
\begin{align}\label{dfttf1}
    \begin{split}
        0 &= 2 (1 - \nu^2) \|T_0 T_{t_f}\| \dfrac{d\|T_0 T_{t_f}\|}{d\gamma^T}  \\
        &\hspace{0.5cm} + 2\Big[-d^{ST}_0 \sin (\gamma^T - \theta^{ST}_0) \|T_0 T_{t_f}\| \\
        &\hspace{1cm} + \big(d^{ST}_0 \cos(\gamma^T - \theta^{ST}_0) - \nu R\big) \dfrac{d\|T_0 T_{t_f}\|}{d\gamma^T}\Big].
    \end{split}
    \end{align}
    Rearranging terms, we solve for $\dfrac{d\|T_0 T_{t_f}\|}{d\gamma^T}$ to obtain
    \begin{align}\label{dfttf2}
        \dfrac{d\|T_0 T_{t_f}\|}{d\gamma^T} = \dfrac{d^{ST}_0 \sin(\gamma^T - \theta^{ST}_0) \|T_0 T_{t_f}\|}{(1 - \nu^2) \|T_0 T_{t_f}\| + d^{ST}_0 \cos(\gamma^T - \theta^{ST}_0) - \nu R}.
    \end{align}
    Setting $\dfrac{d\|T_0 T_{t_f}\|}{d\gamma^T} = 0$ for minimization leads to
    \begin{align}
        d^{ST}_0 \sin(\gamma^T - \theta^{ST}_0) \|T_0 T_{t_f}\| = 0.
    \end{align}
    This yields the following possible conditions,
    \begin{align}
        d^{ST}_0 = 0, \quad \sin(\gamma^T - \theta^{ST}_0) = 0, \quad \text{or} \quad \|T_0 T_{t_f}\| = 0.
    \end{align}
    The non-trivial solution that minimizes $\|T_0 T_{t_f}\|$ is $\sin(\gamma^T - \theta^{ST}_0) = 0$, which gives:
    \begin{align}\label{headminT}
        \gamma^T = \theta^{ST}_0.
    \end{align}
    Substituting $\gamma^T = \theta^{ST}_0$ into Eq.~\eqref{ttf} leads to the minimum distance $T$ needs to travel in order to escape $R$. Thus,
\begin{align}\label{ttfmin}
        \|T_0 T_{t_f}\|_{\min} &= \dfrac{R - d^{ST}_0}{1 - \nu},
    \end{align}
    and therefore, the minimum escape time is $t_f^{\min} =  \|T_0 T_{t_f}\|_{\min}/v^T $ which leads to Eq. (\ref{tfmin}).
\end{proof}
\medskip
So far, we have discussed the optimal sensor strategy and analyzed the interactions between the sensor and the target -- specifically, how the sensor attempts to keep the target within range while the target tries to evade detection. In the next subsection, we present the derivation of the attacker's strategy for capturing the target before it escapes.
\subsection{Strategy for $A$}
The optimal strategy for $A$ is to capture the target $T$ while it remains within the sensing radius $R$ of the sensor $S$. Consider the engagement between $A$ and $T$. Let us start by defining the Apollonius circle between the agents.

\medskip
\begin{definition}[Apollonius Circle between $A$ and $T$]\label{def:ap}
The \emph{Apollonius circle} between the attacker $A$ and the target $T$ at time $t$ is defined as the set of all points $\hat{P} \in \mathbb{R}^2$ in the plane for which the ratio of the distances from $A_t$ and from $T_t$ equals the speed ratio $\mu = \dfrac{v^T}{v^A}$. Denoting this set by $\mathcal{A}_t$, we write:
\begin{align}\label{Apollonius_AT} 
\mathcal{A}_t = \left\{ \hat{P} {\in \mathbb{R}^2} \,\middle|\, \dfrac{\|\hat{P} A_t\|}{\|\hat{P} T_t\|} = \mu \right\}. 
\end{align}\flushright $\blacksquare$ 
\end{definition}

The center and radius of the Apollonius circle depend on the speed ratio $\mu$ and the instantaneous separation between $A$ and $T$. Let $\mathcal{C}^{Ap}_t$ and $R^{Ap}_t$ denote the center and radius of $\mathcal{A}_t$ at time $t$. Then, {it follows from \cite{weintraub2020optimal}, that}
\begin{align} 
\mathcal{C}^{Ap}_t &= \left[T^x_t + \dfrac{(T^x_t - A^x_t) \mu^2}{1 - \mu^2},\, T^y_t + \dfrac{(T^y_t - A^y_t) \mu^2}{1 - \mu^2} \right], \label{cAP} \\
R^{Ap}_t &= \dfrac{\mu}{(1 - \mu^2)} d^{AT}_t. \label{Rap}
\end{align}
 
The points contained inside $\mathcal{A}_t$ correspond to those locations that $T$ can reach before $A$, while the points outside the circle indicate regions where $A$ arrives first. The {boundary} of the circle defines the set of points {that can be reached simultaneously by both agents}. 
Figure~\ref{fig:dirA} illustrates the Apollonius circle between $A$ and $T$ at $t = 0$, centered at $\mathcal{C}^{Ap}_0$ and with radius $R^{Ap}_0$.

For the attacker $A$ to intercept the target $T$ at its final position $T_{t_f}$ in minimum time, $T_{t_f}$ must lie on the Apollonius circle $\mathcal{A}_{t_f}$ and be {inside} the sensable region $\mathcal{S}_{t_f}$.

This condition ensures that both $A$ and $T$ reach $T_{t_f}$ simultaneously. Therefore, the following relation must hold
\begin{align}\label{ap}
\dfrac{\|A_0 T_{t_f}\|}{v^A} = \dfrac{\|T_0 T_{t_f}\|}{v^T} \implies \dfrac{\|T_0 T_{t_f}\|}{\|A_0 T_{t_f}\|} = \dfrac{v^T}{v^A} = \mu. \end{align}
Note that, from the definition of $\mathcal{A}_t$, the optimal direction for the attacker, say $\gamma^{A\star}$,  to intercept the target in minimum time is achieved by heading towards the final position of the target ($T_{t_f}$), at the boundary of $\mathcal{A}_t$.
         \begin{figure}[!h]
        \centering   \includegraphics[width=0.85\linewidth]{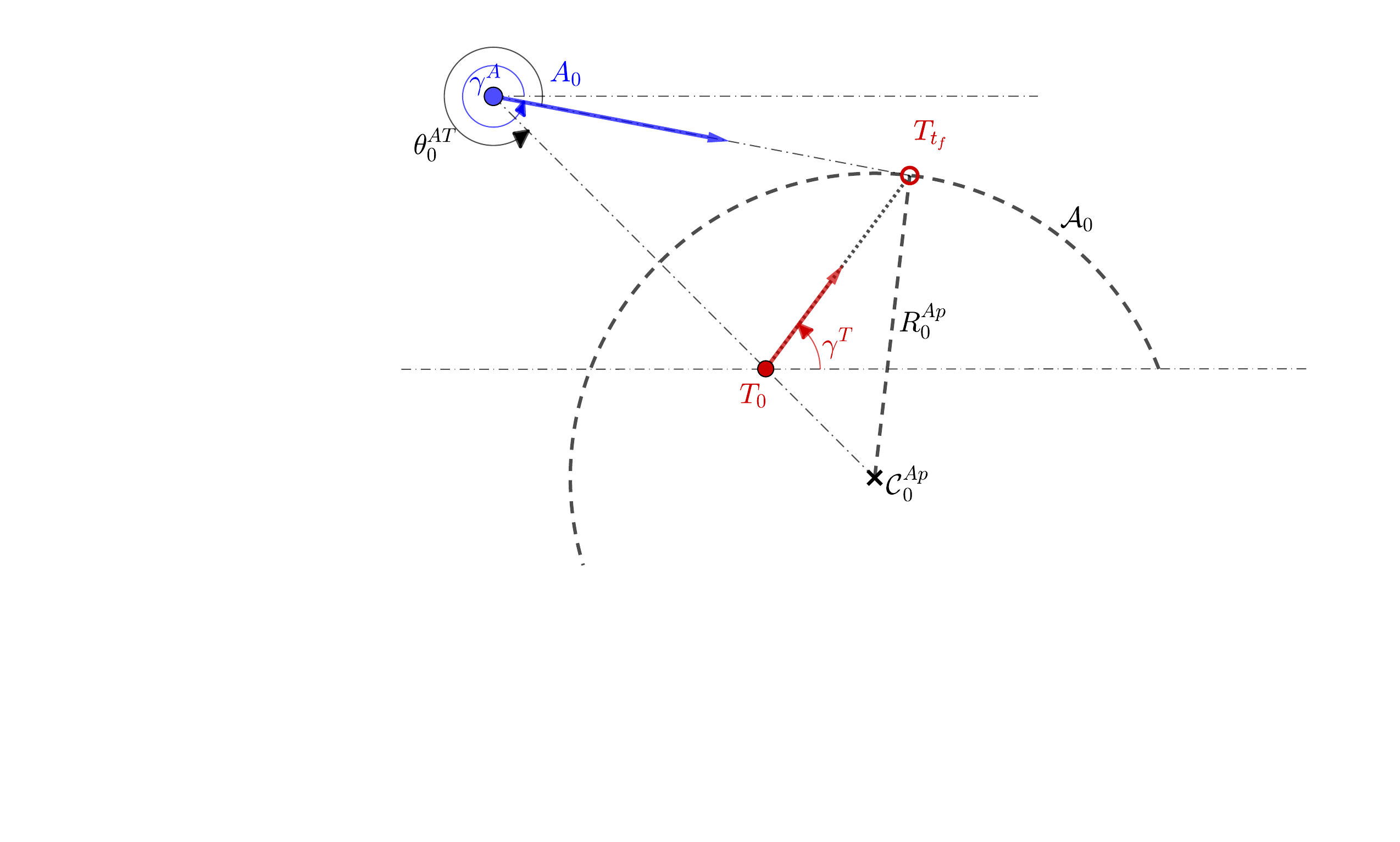}
        \caption{Attacker-target engagement with Apollonius circle.}
        \label{fig:dirA}
    \end{figure}
Applying the law of sines to $\triangle A_0 T_0 T_{t_f}$, as shown in Fig. \ref{fig:dirA}, we have:
\begin{align}\label{ap1}
    \dfrac{\|A_0 T_{t_f}\|}{\sin\left(\pi - (\theta^{AT}_0 - \gamma^T)\right)} = \dfrac{\|T_0 T_{t_f}\|}{\sin(\theta^{AT}_0 - \gamma^{A\star})}.
\end{align}
Simplifying Eq.~\eqref{ap1} and using Eq.~\eqref{ap}, we obtain
\begin{align}\label{sgmA}
    \sin(\theta^{AT}_0 - \gamma^{A\star}) = \mu \sin(\theta^{AT}_0 - \gamma^T).
\end{align}
Rearranging the above yields the optimal heading for $A$ as
\begin{align}\label{gammaA}
    \gamma^{A\star} = \theta^{AT}_0- \sin^{-1}\left[\mu \sin(\theta^{AT}_0 - \gamma^T)\right].
\end{align}
Thus, Eq. (\ref{gammaA}) provides the attacker's heading that ensures target interception in minimum time, provided the target remains within $\mathcal{S}_t$, for all $t \geq 0$.

Having presented the optimal strategies and their key properties for both the attacker and the sensor, we are now prepared to address the solution to the game of kind, in the following section.

\section{Main results}\label{SecIV}
In this section, we establish the necessary and sufficient conditions for target capture, as well as a set of conditions -- parameterized by the initial geometry among the agents -- that determine whether the team of attacker and sensor wins or the target wins.
\begin{lemma}[Capture Condition]\label{thmA0S0}
Suppose that the sensor-attacker team follows the strategies $\gamma^{S\star}$ and $\gamma^{A\star}$, given by Eqs. (\ref{eq:gammastar}) and (\ref{gammaA}), respectively. Then, a necessary and sufficient condition for guaranteed target capture is that $\mathcal{A}_0 \subset \mathcal{S}_0$.
\end{lemma}
\begin{proof}[Proof of Lemma \ref{thmA0S0}]
From Definition \ref{def:ap}, we know that any potential capture point must lie on the boundary of the Apollonius circle $\mathcal{A}_t$. At the initial time $t = 0$, the target $T$ can select any constant heading in an attempt to escape the sensable region. However, if the entire Apollonius circle $\mathcal{A}_0$ is contained within the sensable region $\mathcal{S}_0$, then for any value of $\gamma^T$ chosen by $T$, it cannot avoid detection before being captured. In this case, the target cannot reach any point outside $\mathcal{S}_0$ faster than the attacker. Therefore, the attacker, by following its optimal strategy $\gamma^{A\star}$, is guaranteed to capture the target.

Conversely, suppose there exists a point $\hat{P} \in \mathcal{A}_0$ that lies outside the sensable region $\mathcal{S}_0$. Then, {there exists a point $T_{t_f}$ on the boundary of $\mathcal{S}_0$ such that $T_0, T_{t_f}$ and $\hat{P}$ are collinear and $\|T_0T_{t_f}\| < \|T_0\hat{P}\|$. In this case, the target can adopt the strategy $\bar{\gamma}^T := v^T(\hat{P} - T_0)/\|\hat{P} - T_0\|$, which ensures it reaches $T_{t_f}$ before it gets intercepted}. Therefore, the condition $\mathcal{A}_0 \subset \mathcal{S}_0$ is both necessary and sufficient for guaranteed capture under the proposed strategies $\gamma^{S\star}$ and $\gamma^{A\star}$.
\end{proof}

\medskip 

While Lemma \ref{thmA0S0} gives a necessary and sufficient condition for capture, it does not prescribe any explicit conditions on the problem parameters that determine the winning outcomes.  {However, Lemma \ref{thmA0S0} prescribes an escape strategy for the target whenever it exists. } {Moreover, Lemma \ref{thmA0S0} plays a crucial role in establishing the subsequent results that define conditions on the target speed $v^T$ under which capture or escape is guaranteed. }
 
\begin{theorem}[{Sensor-Attacker Winning Condition}]\label{thmVT}
   Suppose that the sensor and the attacker move as per the strategies $\gamma^{S\star}$ and $\gamma^{A\star}$ defined in Eqs. (\ref{eq:gammastar}) and (\ref{gammaA}), respectively, where the target speed $v^T$, sensor speed $v^S$, and the attacker speed $v^A$ satisfy Assumption~\ref{as:speeds} and the sensing radius $R$ exceeds the initial distance $d^{ST}_0$ between the sensor and the target. 
    If
 \begin{align}\label{vTmax}
     v^T < \underline{v}^T, 
 ~\text{where}~ \underline{v}^T = \dfrac{(R - d^{ST}_0) v^A + d^{AT}_0 v^S}{d^{AT}_0 + R - d^{ST}_0 },
    \end{align}
   then capture is guaranteed.
   Additionally, if Eq.~\eqref{vTmax} does not hold, then there exist initial locations for the agents and a strategy $\gamma^T$ that leads to escape for $T$.
\end{theorem}
\begin{remark}[Target Objectives in Case of No Escape]
If the initial conditions are such that the target $T$ cannot escape, the target’s strategy depends upon whether it aims to end the engagement quickly or delay interception. To minimize the time to capture, the target should move directly toward the attacker, that is,  $\gamma^T = -\theta^{AT}_0$. Conversely, to maximize the time until interception, the target should move away from the attacker, leading to $\gamma^T = \theta^{AT}_0$. 
\end{remark}

\begin{proof}[Proof of Theorem \ref{thmVT}] Consider a typical engagement scenario among the agents $S$, $A$, and $T$, as shown in Fig.~\ref{fig:capture}. The corresponding sensable region $\mathcal{S}_0$ -- defined with respect to $S_0$ and $T_0$ -- and the Apollonius circle $\mathcal{A}_0$, constructed based on $A_0$ and $T_0$, are also illustrated. 
 
Assume that, at $t = 0$, the target's speed satisfies Eq. (\ref{vTmax}).
This inequality can be rearranged as
\begin{align*}
    v^T (d^{AT}_0 + R - d^{ST}_0) &- (R - d^{ST}_0) v^A + d^{AT}_0 v^S < 0.
\end{align*}
Simplifying further, we obtain
\begin{align*}
    \dfrac{v^T}{v^A} d^{AT}_0 \left(1 - \dfrac{v^S}{v^T} \right) - (R - d^{ST}_0) \left(1 - \dfrac{v^T}{v^A} \right) < 0.
\end{align*}
Substituting $\mu = \dfrac{v^T}{v^A}$ and $\nu = \dfrac{v^S}{v^T}$ with some algebraic manipulations lead to
\begin{align}
    \label{65}
    \dfrac{\mu}{1 - \mu} d^{AT}_0 - \dfrac{R - d^{ST}_0}{1 - \nu} < 0.
\end{align}
According to Corollary~\ref{cor1}, the minimum distance from the target's initial position $T_0$ to a possible escape point $T_{t_f}$ on the boundary of $\mathcal{S}_0$ is given by Eq. (\ref{ttfmin}).

Next, we rewrite the first term in Eq.~(\ref{65}) as
\begin{align*}
    \dfrac{\mu }{1 - \mu}d^{AT}_0 = \left(\dfrac{\mu}{1 - \mu^2} + \dfrac{\mu^2}{1 - \mu^2} \right) d^{AT}_0.
\end{align*}
From Eq.~(\ref{cAP}), we can calculate the distance from the center of the Apollonius circle to the target’s initial position, which is given by
\begin{align}\label{CT0}
    \|\mathcal{C}^{Ap}_0 - T_0\| = \dfrac{\mu^2}{1 - \mu^2} d^{AT}_0.
\end{align}
Substituting Eqs.~(\ref{ttfmin}), (\ref{Rap}), and (\ref{CT0}) into Eq.~(\ref{65}) yields
\begin{align}\label{ApS1}
    R^{Ap}_0 + \|\mathcal{C}^{Ap}_0 - T_0\| < \|T_0 T_{t_f}\|_{\min}.
\end{align}
\begin{figure}[!h]
        \centering
\includegraphics[width=\linewidth]{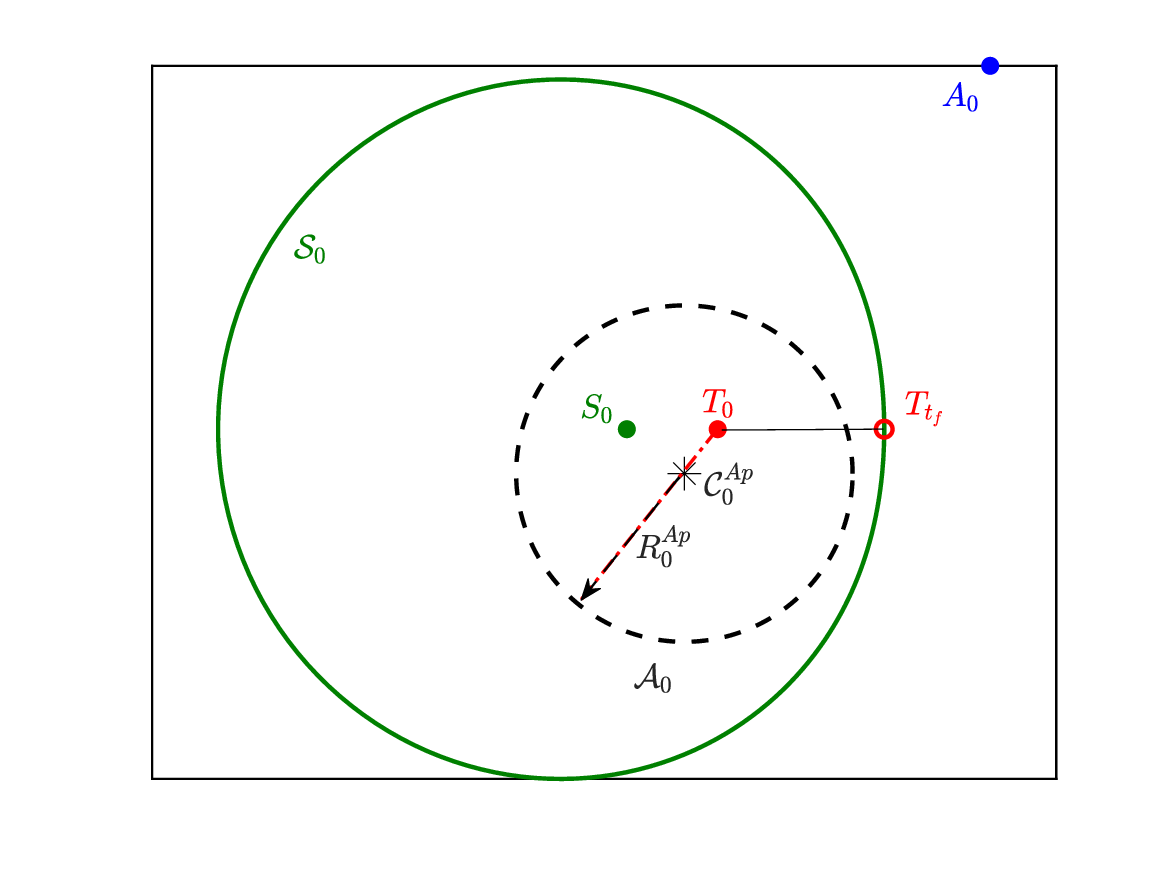}
           \caption{A typical engagement scenario between the three agents, showing $\mathcal{A}_0 \subset \mathcal{S}_0.$}
        \label{fig:capture}
    \end{figure}
Equation~(\ref{ApS1}) implies that the maximum distance from $T_0$ to any point on the boundary of $\mathcal{A}_0$ is less than the minimum distance between $T_0$ and the boundary of $\mathcal{S}_0$, as shown in Fig. \ref{fig:capture}. In other words, the entire Apollonius region $\mathcal{A}_0$ lies strictly within the sensable region $\mathcal{S}_0$, that is $\mathcal{A}_0 \subset \mathcal{S}_0$.
Lemma~\ref{thmA0S0} completes the proof.
\end{proof}
\medskip 

Although we have established a sufficient condition that guarantees target capture, it is equally important to characterize parameter regimes where the escape of the target is guaranteed. The following result presents a condition on the target speed $v^T$ beyond which the target is guaranteed to escape despite the optimal strategies of the sensor-attacker team. 
\begin{theorem}[{Target Winning Condition}]\label{thmesc}
 Suppose the sensor and attacker follow the strategies $\gamma^{S\star}$ and $\gamma^{A\star}$ defined in Eqs.~(\ref{eq:gammastar}) and (\ref{gammaA}), respectively. Let the target speed $v^T$, sensor speed $v^S$, and attacker speed $v^A$ satisfy Assumption~\ref{as:speeds}, and assume the sensing radius $R$ exceeds the initial sensor-target distance $d_0^{ST}$. If
 \begin{align}\label{vTmin}
  v^T > \overline{v}^T,
    \text{where}~\overline{v}^T &= \dfrac{(R - d_0^{ST}) v^A + d_0^{AT} v^S}{d_0^{AT} - R + d_0^{ST}},
\end{align}
then there exists a target strategy $\gamma^T$ that guarantees escape.
\end{theorem}
\begin{remark}[Admissible Values of $\overline{v}^T$]
It is important to note that the target speed $v^T$ must also comply with Assumption \ref{as:speeds}, even when the target is attempting to escape.  Consequently, the inequality in Eq. (\ref{vTmin}) is meaningful only when the following condition holds
\begin{align}
   \dfrac{(R - d_0^{ST}) v^A + d_0^{AT} v^S}{d_0^{AT} - R + d_0^{ST}} < v^A. 
\end{align}
This condition corresponds to a regime in which the initial positions of the agents and the relative speed ratios satisfy the inequality $2(R - d_0^{{ST}})/d_0^{AT} < 1 - \nu \mu$. This ensures that the escape condition remains consistent with the bounded speed constraints imposed by the problem formulation.
\end{remark}
\begin{remark}[Difference between $\underline{v}^T$ and $\overline{v}^T$]
    Note that the right hand sides of Eq.~\eqref{vTmin} and Eq.~\eqref{vTmax} are very similar except for a sign difference in the denominator. Specifically, the right hand side of Eq.~\eqref{vTmin} is always higher than that of Eq.~\eqref{vTmax}. 
\end{remark}
\begin{figure}[!h]
    \centering
    \includegraphics[width=\linewidth]{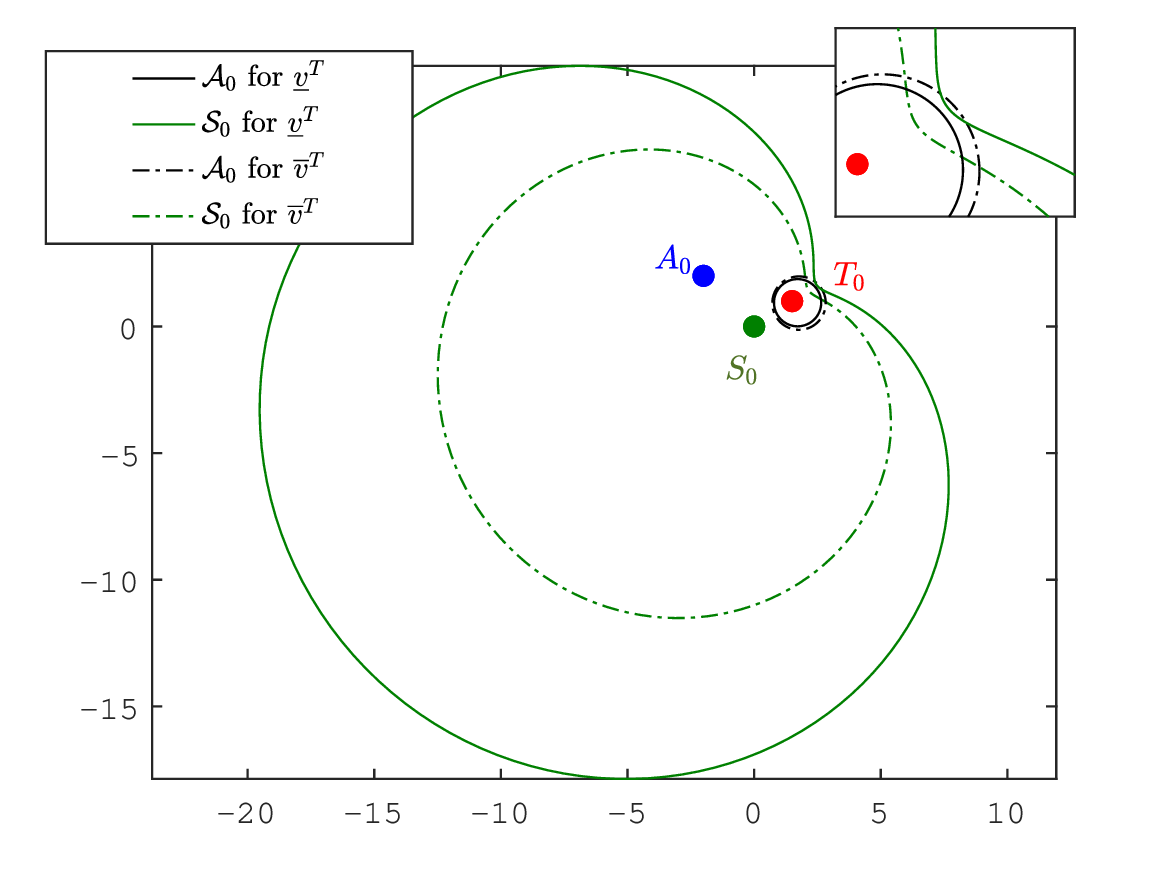}
    \caption{Visualization of the gap between $\underline{v}^T$ and $\overline{v}^T$ for the given initial positions of the sensor, attacker and the target.}
    \label{fig:remark3}
\end{figure}

\medskip
\begin{proof}[Proof of Theorem \ref{thmesc}]
\begin{figure}[!h]
        \centering
\includegraphics[width=\linewidth]{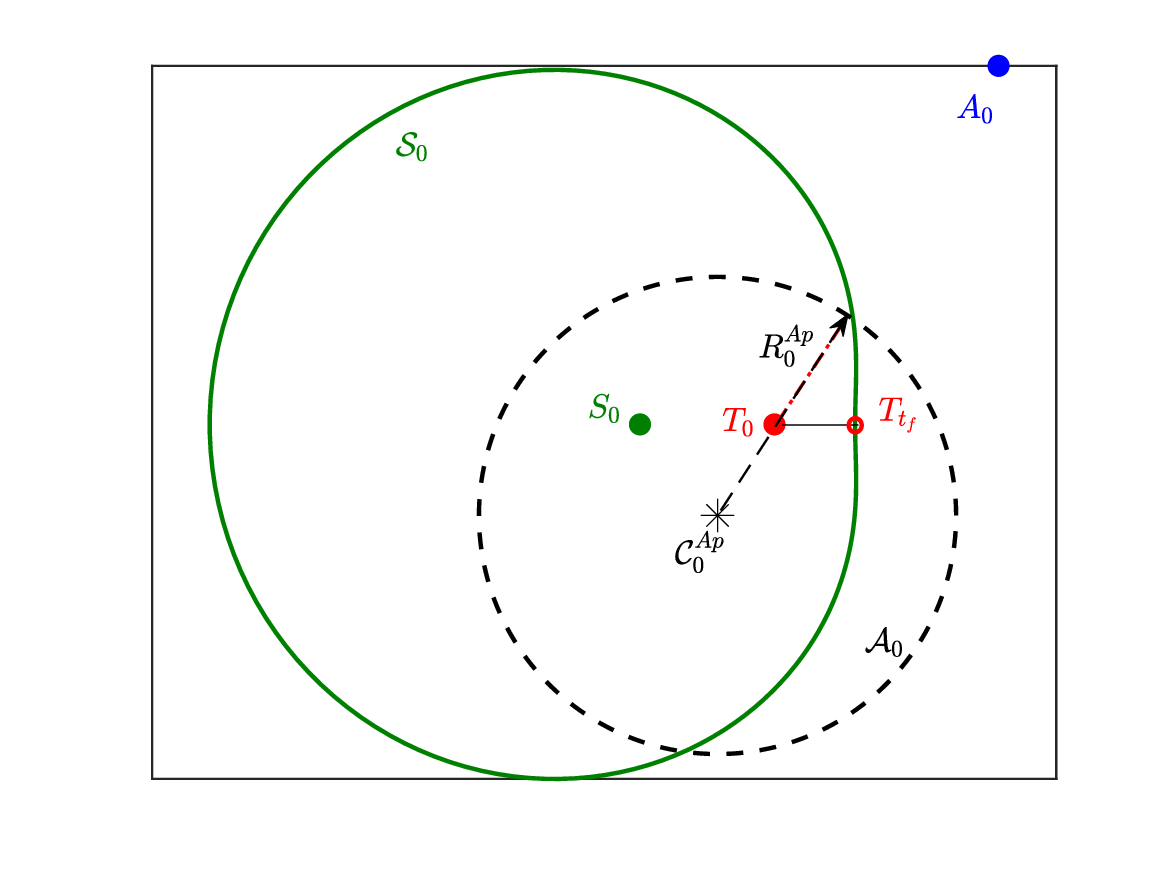}
           \caption{A typical engagement scenario between the three agents, showing $\mathcal{A}_0 \not \subset \mathcal{S}_0$.}
        \label{fig:escape}
    \end{figure}
   Assume that the initial conditions satisfy Eq.~\eqref{vTmin}. Rearranging the inequality gives
    \begin{align}
        v^T (d_0^{AT} - R + d_0^{ST}) > (R - d_0^{ST}) v^A + d_0^{AT} v^S.
    \end{align}
    Proceeding with steps analogous to those used in the proof of Theorem~\ref{thmVT}, we obtain
    \begin{align}\label{75}
        \dfrac{\mu}{1+\mu} d_0^{AT} - \dfrac{R - d_0^{ST}}{1 - \nu} > 0
    \end{align}
    The first term in the above inequality can be re-written as
    \begin{align}\label{ftm}
        \dfrac{\mu}{1 + \mu} d_0^{AT} = \left(\dfrac{\mu}{1 - \mu^2} - \dfrac{\mu^2}{1 - \mu^2}\right)d_0^{AT}.
    \end{align}
Using Eqs.~(\ref{ttfmin}), (\ref{Rap}), (\ref{CT0}), and (\ref{ftm}), we substitute into Eq.~\eqref{75} to obtain
    \begin{align}
        R^{Ap}_0 - \|\mathcal{C}^{Ap}_0 - T_0 \| > \|T_0 T_{tf}\|_{\min}.
    \end{align}
  Geometrically, this implies that the minimum distance from $T_0$ to any point on the boundary of the Apollonius circle $\mathcal{A}_0$ exceeds the minimum distance from $T_0$ to the boundary of the sensable region $\mathcal{S}_0$. In other words, the Apollonius circle extends beyond the sensable region at $t = 0$, as illustrated in Fig.~\ref{fig:escape}. This violates the necessary and sufficient condition for guaranteed capture established in Lemma~\ref{thmA0S0}. Consequently, there exists at least one escape trajectory for the target, {as described in the proof of Lemma \ref{thmA0S0}}, under these initial conditions.
\end{proof}

\medskip 
The bounds on the target speed from Theorems \ref{thmVT} and \ref{thmesc} apply to \emph{any} initial values of the LOS angles among the agents. {If we restrict our attention to given initial values of the LOS angles between the agents, then we can further analyze the critical target speed. In particular, our next result Theorem \ref{thm:vTcrit}} proves the existence of a critical target speed that defines the boundary between successful capture or evasion.

\medskip

\medskip
\begin{remark}[Sensitivity to Localization Errors]
    The speed thresholds $\overline{v}^T$ and $\underline{v}^T$ depend on the initial distances $d_0^{ST}$ and $d_0^{AT}$. A first-order sensitivity analysis shows that small overestimation errors in these measurements reduce both thresholds, yielding closed-form but conservative estimates that are beneficial for robust real-time implementation. 
\end{remark}

\begin{theorem}[Existence of Critical Target Speed]\label{thm:vTcrit}
    Given the initial conditions, if the sensor and attacker follow the strategies specified in Eqs.~(\ref{eq:gammastar}) and (\ref{gammaA}), respectively, and if the sensor speed $v^S$, target speed $v^T$, and attacker speed $v^A$ satisfy Assumption~\ref{as:speeds}, then there exists a critical speed $v^{T\star}$ within the interval
    \begin{align}
        \underline{v}^T \leq v^{T\star} \leq \overline{v}^T,
    \end{align}
    where $\underline{v}^T$ and $\overline{v}^T$ are given in Eqs. (\ref{vTmax}) and (\ref{vTmin}), respectively,  such that
\begin{enumerate}
    \item If $v^T > v^{T\star}$, the target successfully escapes ($T$ wins).
    \item If $v^T \leq v^{T\star}$, capture is guaranteed (the team consisting of $A$ and $S$ wins).
\end{enumerate}

\end{theorem}
{\begin{remark}[Multiple Tangency Points Scenario]
 Due to the non-convexity of the sensable region, the Apollonius circle and sensable region may be tangent at multiple points, each potentially corresponding to a different critical speed $v^{T\star}$. However, for any such critical speed $v^{T\star}$, increasing the target speed beyond this value ensures that the two regions intersect. The smallest $v^T$ at which the two curves are tangent to each other is then considered as the critical speed $v^{T\star}$. Any $v^T > v^{T\star}$ results in the intersection of the two curves leading to the escape direction for the target. Therefore, although multiple tangency points may exist, the condition $v^T > v^{T\star}$ still guarantees the target’s victory.
\end{remark}}
\begin{proof}[Proof of Theorem \ref{thm:vTcrit}]
From the definitions of the sensable region $\mathcal{S}_0$ and the Apollonius circle $\mathcal{A}_0$, given in Definition~\ref{def:sensable} and Definition~\ref{def:ap}, respectively, it is evident that both $\mathcal{S}_0(v^T)$ and $\mathcal{A}_0(v^T)$ vary continuously with $v^T$. Furthermore, since $\mathcal{S}_0(v^T)$ characterizes the distance the target $T$ must travel to escape sensing, for fixed values of $v^A$ and $v^S$, an increase in $v^T$ results in a contraction of $\mathcal{S}_0(v^T)$. Conversely, because $\mathcal{A}_0(v^T)$ depends on the relative separation between the attacker $A$ and the target $T$, an increase in $v^T$ leads to an expansion of $\mathcal{A}_0(v^T)$.
\begin{figure}[!h]
    \centering
    \includegraphics[width=\linewidth]{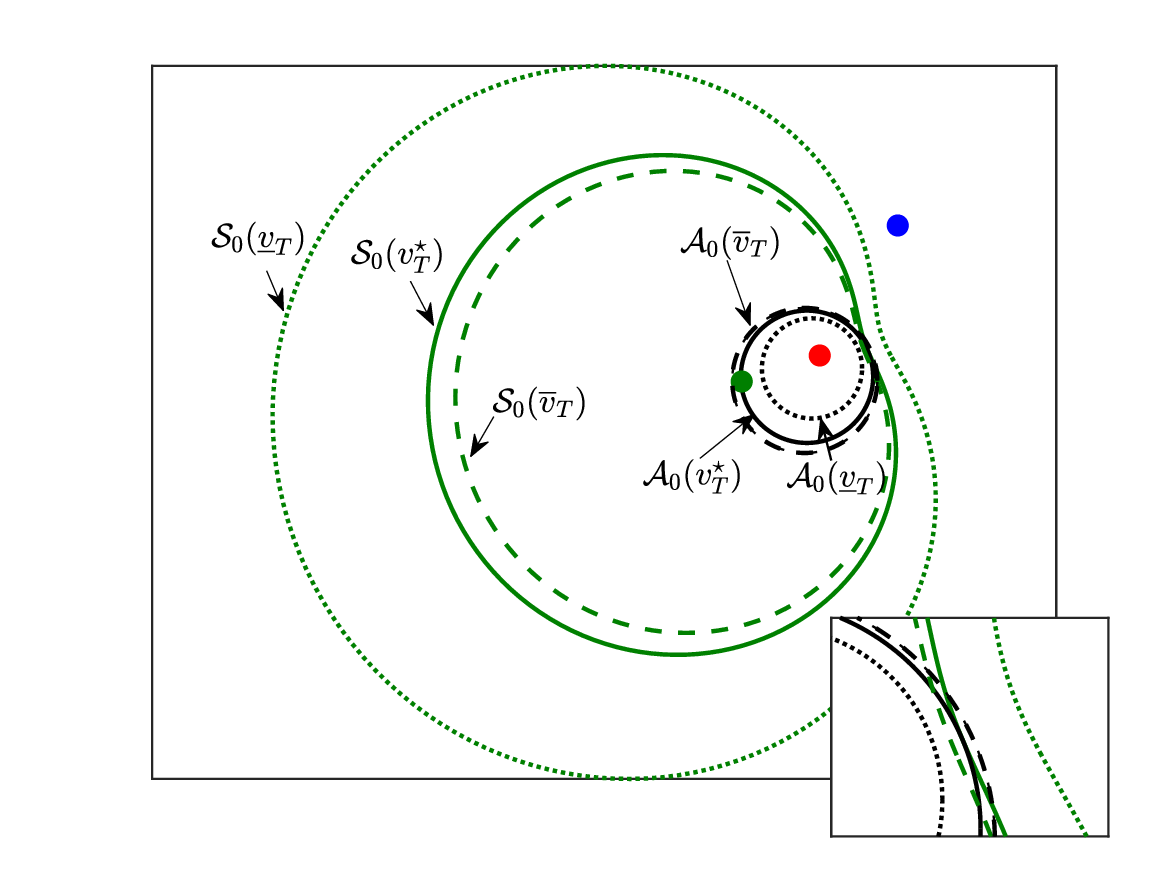}
    \caption{Visualization of $\mathcal{S}_0$ and $\mathcal{A}_0$ for target velocity values $\underline{v}^T$ (dotted line), $v^{T\star}$ (solid line), and $\overline{v}^T$ (dashed line).}
    \label{fig:vTcexist}
\end{figure}
From Theorem~\ref{thmVT} and Theorem~\ref{thmesc}, we have established that for $v^T \leq \underline{v}^T$, the set $\mathcal{A}_0(v^T)$ remains fully contained within $\mathcal{S}_0(v^T)$, guaranteeing capture. However, for $v^T \geq \overline{v}^T$, $\mathcal{A}_0(v^T)$ intersects with $\mathcal{S}_0(v^T)$, allowing the target to escape.  Therefore, in the interval $(\underline{v}^T,\,\overline{v}^T)$, there must exist a critical target speed $v^{T\star}$ at which $\mathcal{A}_0(v^T)$ is tangent to $\mathcal{S}_0(v^T)$. Specifically, the following relations hold
\begin{align*}
    \mathcal{S}_0(\overline{v}^{T}) \subset  \mathcal{S}_0(v^{T\star}) \subset  \mathcal{S}_0(\underline{v}^T), \\
    \mathcal{A}_0(\underline{v}^T) \subset  \mathcal{A}_0(v^{T\star}) \subset  \mathcal{A}_0(\overline{v}^T).
\end{align*}
This is illustrated in Fig.~\ref{fig:vTcexist}, which depicts the evolution of $\mathcal{S}_0(v^T)$ and $\mathcal{A}_0(v^T)$ as $v^T$ increases, thereby confirming the existence of the critical speed $v^{T\star}$.
\end{proof}
\medskip 
It is important to note that the point at which $\mathcal{A}_0$ and $\mathcal{S}_0$ are tangent to each other is not necessarily the point corresponding to the minimum distance required by the target to escape, as given in Eq. (\ref{ttfmin}). However, under the specified initial conditions, by assuming that the tangency point coincides with the point of minimum escape distance, we can refine the upper bound $\overline{v}^T$ from Theorem~\ref{thmesc} to more precisely characterize the speed regime that ensures the target's escape. We formalize this in the following result.

\begin{theorem}[Escape Condition - Given Initial Geometry]\label{thm:vTstar_tangent}
   Suppose that the sensor and the attacker follow the strategies in Eqs.~(\ref{eq:gammastar}) and (\ref{gammaA}), respectively, and the sensor speed $v^S$, target speed $v^T$, and attacker speed $v^A$ satisfy Assumption~\ref{as:speeds}. Then there exists a speed $\overline{v}^{T\star} < \overline{v}^T$, where the value of $\overline{v}^T$ is given in Eq. (\ref{vTmin}), such that 
     if
     $$v^T > \overline{v}^{T\star},$$ 
     then there exists an escape strategy for the target ($T$ wins). Further, the value of $\overline{v}^{T\star}$ satisfies the quadratic equation
\begin{align} \label{vTcrit}
a {\overline{v}^{T\star}}^2 + b \overline{v}^{T\star} + c = 0,~~~~
\end{align}
where
\begin{align*}
   a &=  1 - 2 \left(\dfrac{R - d^{ST}_0}{d^{AT}_0}\right) \cos (\theta^{AT}_0 - \theta^{ST}_0) + \left(\dfrac{R - d^{ST}_0}{d^{AT}_0}\right)^2, \\
b &= - 2 v^S \left[( 1 - \left(\dfrac{R - d^{ST}_0}{d^{AT}_0}\right) \cos  (\theta^{AT}_0 - \theta^{ST}_0)\right], \\ 
c &= {v^{S}}^2 -\left(\dfrac{R - d^{ST}_0}{d^{AT}_0}\right)^2 {v^{A}}^2,
\end{align*}
with $d_0^{ST}$ and $d_0^{AT}$ denoting the initial separations between the sensor and the target, and between the attacker and the target, respectively, $\theta_0^{AT}$ denoting the initial line-of-sight angle between the attacker and the target, $\theta_0^{ST}$ denoting the initial line-of-sight angle between the sensor and the target, and $R$ being the sensing radius of the sensor.
\end{theorem}
\begin{proof}[Proof of Theorem \ref{thm:vTstar_tangent}]
Consider the initial engagement geometry as depicted in Fig.~\ref{fig:vTcrit}, such that the point of minimum escape distance for the target -- denoted by $T_{t_f}$ lies on $\mathcal{A}_0$.
Let us now assume $v^T = \overline{v}^{T\star}$.
\begin{figure}[!h]
    \centering
    \includegraphics[width=\linewidth]{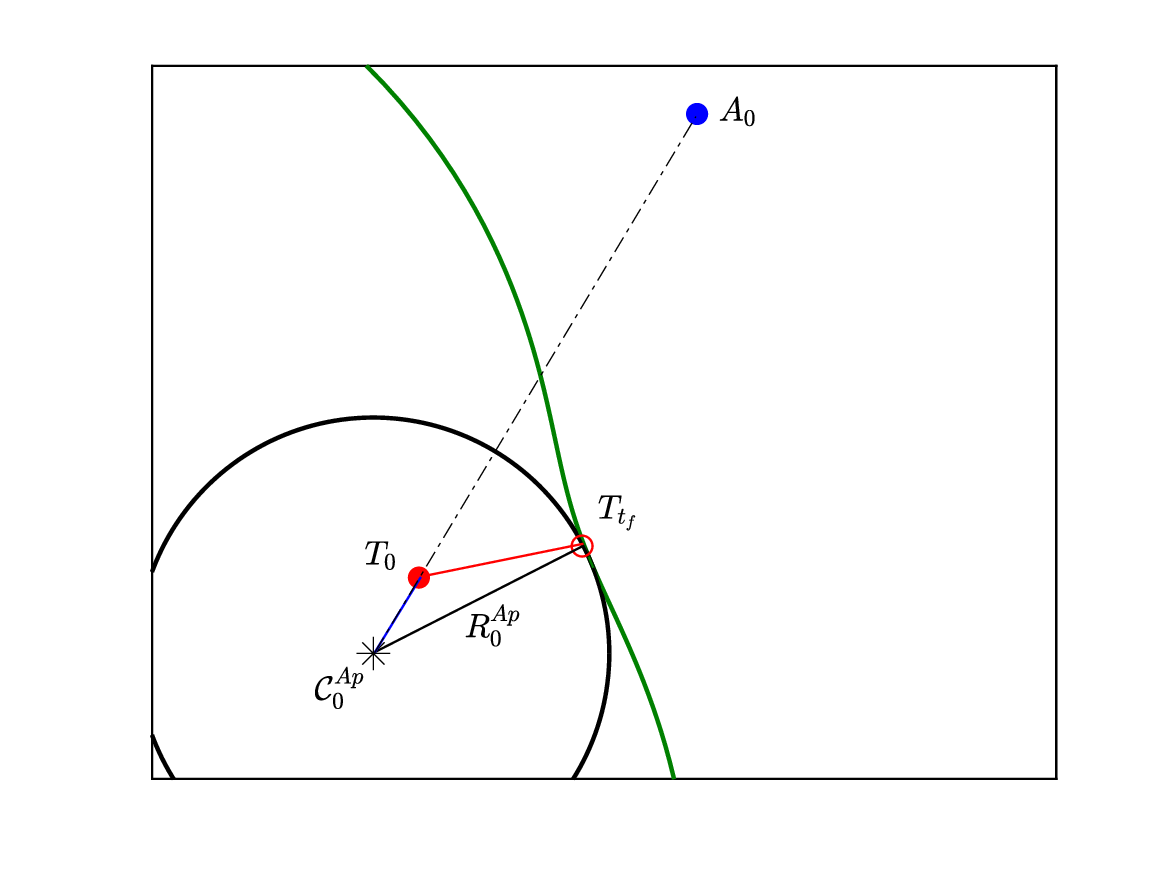}
    \caption{Critical configuration where $\mathcal{A}_0$ and $\mathcal{S}_0$ are tangent.}
    \label{fig:vTcrit}
\end{figure}
Applying the Law of Cosines to triangle $\triangle \mathcal{C}^{Ap}_0 T_0 T_{t_f}$ yields
\begin{align} \nonumber
    {R^{Ap}_0}^2 &= \|\mathcal{C}^{Ap}_0 - T_0\|^2 + \|T_0 T_{tf}\|_{\min}^2 \\ \label{csn}
    &~- 2 \|\mathcal{C}^{Ap}_0 - T_0\| \|T_0 T_{tf}\|_{\min} \cos [\pi - (\theta^{AT}_0 -\gamma^T )].
\end{align}
The minimum escape distance for the target is achieved at $\gamma^T = \theta^{ST}_0$ as proved in Corollary \ref{cor1}.
Using Eqs.~(\ref{Rap}) and~(\ref{CT0}), Eq.~(\ref{csn}) simplifies to
\begin{align}\nonumber
   \|T_0 T_{tf}\|_{\min}^2 + 2 \left(\dfrac{\mu^2}{1-\mu^2}\right) &{d^{AT}_0} \|T_0 T_{tf}\|_{\min} \cos (\theta^{AT}_0 -\theta^{ST}_0 ) \\& \quad - \left( \dfrac{\mu^2}{1-\mu^2}\right) {d^{AT}_0}^2 = 0.
\end{align}
The solution for $\|T_0 T_{t_f}\|_{\min}$ is
\begin{align}\nonumber
   \|T_0 T_{t_f}\|_{\min} &= \dfrac{\mu}{1-\mu^2}d^{AT}_0 \left(-\mu \cos (\theta^{AT}_0 -\theta^{ST}_0 ) \right. \\ \label{rt} & \qquad  \qquad \left. \pm \sqrt{1 - \mu^2 \sin^2 (\theta^{AT}_0 -\theta^{ST}_0 )}\right) 
\end{align}
Substituting $\|T_0 T_{t_f}\|_{\min}$ from Eq.~(\ref{ttfmin}) into the above expression leads to
\begin{align} \nonumber
    \dfrac{R - d^{ST}_0}{d^{AT}_0} \left(\dfrac{1-\mu^2}{\mu (1 - \nu)}\right)  = &- \mu \cos (\theta^{AT}_0 -\theta^{ST}_0 ) \\ \label{sim}
    &\pm \sqrt{1 - \mu^2 \sin ^2 (\theta^{AT}_0 -\theta^{ST}_0)}. 
\end{align}
For the convenience of notations, let us denote the constants in the equation as  $D = \dfrac{R - d^{ST}_0}{d^{AT}_0}$, $P_\theta = \cos (\theta^{AT}_0 -\theta^{ST}_0 )$ and $Q_\theta = \sin (\theta^{AT}_0 -\theta^{ST}_0 )$. Using these notations, Eq. (\ref{sim}) can be represented as
\begin{align}
    D \left(\dfrac{1-\mu^2}{\mu (1 - \nu)}\right) + \mu P_\theta = \pm \sqrt{1 - \mu^2 Q_\theta^2}.
\end{align}
Substituting $\mu = \overline{v}^{T\star}/v^A$ and $\nu = v^S/\overline{v}^{T\star}$, we get
\begin{align}
    D \left(\dfrac{{v^{A}}^2 - {\overline{v}^{T\star}}^2}{v^A(\overline{v}^{T\star} - v^S)}\right) + \dfrac{\overline{v}^{T\star}}{v^A}P_\theta = \pm \dfrac{1}{v^A} \sqrt{{v^{A}}^2 - {\overline{v}^{T\star}}^2 Q_\theta^2},
\end{align}
which simplifies to
\begin{align}
    D \left(\dfrac{{v^{A}}^2 - {\overline{v}^{T\star}}^2}{\overline{v}^{T\star} - v^S}\right) + \overline{v}^{T\star} P_\theta = \pm \sqrt{{v^{A}}^2 - {\overline{v}^{T\star}}^2 Q_\theta^2}.
\end{align}
Squaring both sides yields
\begin{align} \nonumber
    D^2 \left(\dfrac{{v^{A}}^2 - {\overline{v}^{T\star}}^2}{  \overline{v}^{T\star} - v^S}\right)^2 &+ {{\overline{v}^{T\star}}^2}P_\theta^2  \\
    & \hspace{-1cm}+ 2 D \left(\dfrac{{v^{A}}^2 - {\overline{v}^{T\star}}^2}{  \overline{v}^{T\star} - v^S}\right) \overline{v}^{T\star} P_\theta = {v^{A}}^2 - {\overline{v}^{T\star}}^2 Q_\theta^2,\\ \nonumber
    D^2 \left(\dfrac{{v^{A}}^2 - {\overline{v}^{T\star}}^2}{  \overline{v}^{T\star} - v^S}\right)^2 & + 2 D \left(\dfrac{{v^{A}}^2 - {\overline{v}^{T\star}}^2}{  \overline{v}^{T\star} - v^S}\right) \overline{v}^{T\star} P_\theta \\ & \qquad \qquad \qquad + {{\overline{v}^{T\star}}^2} - {v^{A}}^2 = 0.
\end{align}
The above equation can be further simplified as 
\begin{align}
     D^2 ({{v^{A}}^2 - {\overline{v}^{T\star}}^2}) & + 2 DP_\theta{(\overline{v}^{T\star} - v^S)}\overline{v}^{T\star}   - (\overline{v}^{T\star} - v^S)^2 = 0.
\end{align}
Expanding the square terms and simplifying results in a quadratic in $\overline{v}^{T\star}$, such as
\begin{align}
\nonumber
    &\left[ 1 - 2 \left(\dfrac{R - d^{ST}_0}{d^{AT}_0}\right) \cos  (\theta^{AT}_0 - \theta^{ST}_0) + \left(\dfrac{R - d^{ST}_0}{d^{AT}_0}\right)^2 \right]{\overline{v}^{T\star}}^2 \\ \nonumber
 & \quad \qquad -  2 v^S \left[ 1 - \left(\dfrac{R - d^{ST}_0}{d^{AT}_0}\right) \cos  (\theta^{AT}_0 - \theta^{ST}_0)\right] \overline{v}^{T\star} \\
 &\quad  \qquad \quad + {v^{S}}^2 -\left(\dfrac{R - d^{ST}_0}{d^{AT}_0}\right)^2 {v^{A}}^2 = 0.
\end{align}
Therefore, by substituting
\begin{align}\label{coeff}
    a &=  1 - 2 \left(\dfrac{R - d^{ST}_0}{d^{AT}_0}\right) \cos (\theta^{AT}_0 - \theta^{ST}_0) + \left(\dfrac{R - d^{ST}_0}{d^{AT}_0}\right)^2 \\
b &= - 2 v^S \left[ 1 - \left(\dfrac{R - d^{ST}_0}{d^{AT}_0}\right) \cos  (\theta^{AT}_0 - \theta^{ST}_0)\right],\\ 
c &= {v^{S}}^2 -\left(\dfrac{R - d^{ST}_0}{d^{AT}_0}\right)^2 {v^{A}}^2,
\end{align}
the quadratic given in  Eq. (\ref{vTcrit}) is obtained, completing the proof.
\end{proof}
\medskip 
\begin{remark}[Comparison of Theorems \ref{thmesc} and \ref{thm:vTstar_tangent}]
Theorem~\ref{thm:vTstar_tangent} only provides a tighter sufficient condition for the target's escape instead of a necessary and sufficient one. We compare this bound numerically to the one from Theorem~\ref{thmesc}. 
It is also important to note that the result in Theorem~\ref{thmesc} is valid for \emph{any} arbitrary initial orientations between the agents, while the tighter bound presented in Theorem~\ref{thm:vTstar_tangent} is derived specifically for the given set of initial conditions.
\end{remark}

\begin{remark}[Challenges in Defining  Capture Condition]
Theorem \ref{thm:vTstar_tangent} identifies a sufficient condition for escape based on the assumption that the point of tangency between $\mathcal{A}_0$ and $\mathcal{S}_0$ corresponds to the minimum distance the target must travel to exit $\mathcal{S}_0$. However, our numerical results show that tangency can occur at other points, potentially closer than the assumed escape distance.
Analytically determining these points is challenging, as it involves solving two biquadratic equations corresponding to the Apollonius circle and the sensor's sensable region. As a result, when $v^T < \overline{v}^{T\star}$, we cannot definitively conclude capture, since the true critical speed $v^{T\star}$ may be lower. In contrast, if tangency occurs at or beyond the minimum escape distance, intersection is guaranteed, Theorem 2 is violated, and escape is assured. This kind of certainty does not hold for capture, which is why Theorem 5 does not attempt to characterize it.
\end{remark}

\section{Numerical Evaluation and Visualization} \label{SecV}
To visualize the theoretical findings from Section~\ref{SecIII} and Section~\ref{SecIV}, we now present numerical illustrations that demonstrate the behavior of the agents under different initial conditions.

\subsection{Visualization of Theorem \ref{thm1} and Corollary \ref{cor1}}
To visualize Theorem~\ref{thm1} and Corollary~\ref{cor1}, we examine the influence of the target’s heading angle $\gamma^T$ on the engagement’s termination time $t_f$. The initial positions of the agents are defined as $S_0 = (0,\,0)$, $A_0 = (1,\,3)$, and $T_0 = (1.5,\, 1.5)$. The agents move with constant speeds: $v^S = 0.125\,$m/s, $v^T = 0.35\,$m/s, and $v^A = 1\,$m/s. The sensor has a fixed sensing radius of $R = 2\,$m. At the initial time $t = 0$, the relative distances and line-of-sight angles between the agents are given by: $d^{AT}_0 = 1.7678\,$m, $\theta^{AT}_0 = -81.86^\circ$, and $d^{ST}_0 = 1.7678\,$m, with $\theta^{ST}_0 = 45^\circ$.

Figure~\ref{fig:gt_variation} illustrates the trajectories of all three agents for varying values of $\gamma^T$. The plots also display the sensable region $\mathcal{S}_t$ and the Apollonius circle $\mathcal{A}_t$ at the start ($t = 0$, dashed lines) and at the termination time ($t = t_f$, solid lines), highlighting how the geometry of the engagement evolves over time in response to different target strategies.
 \begin{table}[!h]
    \centering
    \begin{tabular}{|c|c|c|}
        \hline
        Target Heading $\left(\gamma^T\right)$ & $\left(\|T_0 T_{t_f}\|\right)$ (m) &  $t_f$ (s) \\
        \hline
        $-81.86^\circ$ & $4.3020$ & $12.2914$ \\
        $45^\circ$ & $0.3613$ & $1.0321$ \\
        $60^\circ$ & $0.3778$ & $1.0794$ \\
        \hline
    \end{tabular}
    \caption{Escape distance and termination time for different values of $\gamma^T$.}
    \label{tab:thm1}
\end{table}
\begin{figure*}[!h]
    \centering
    \begin{subfigure}[t]{0.325\textwidth}
        \centering
\includegraphics[width=\linewidth]{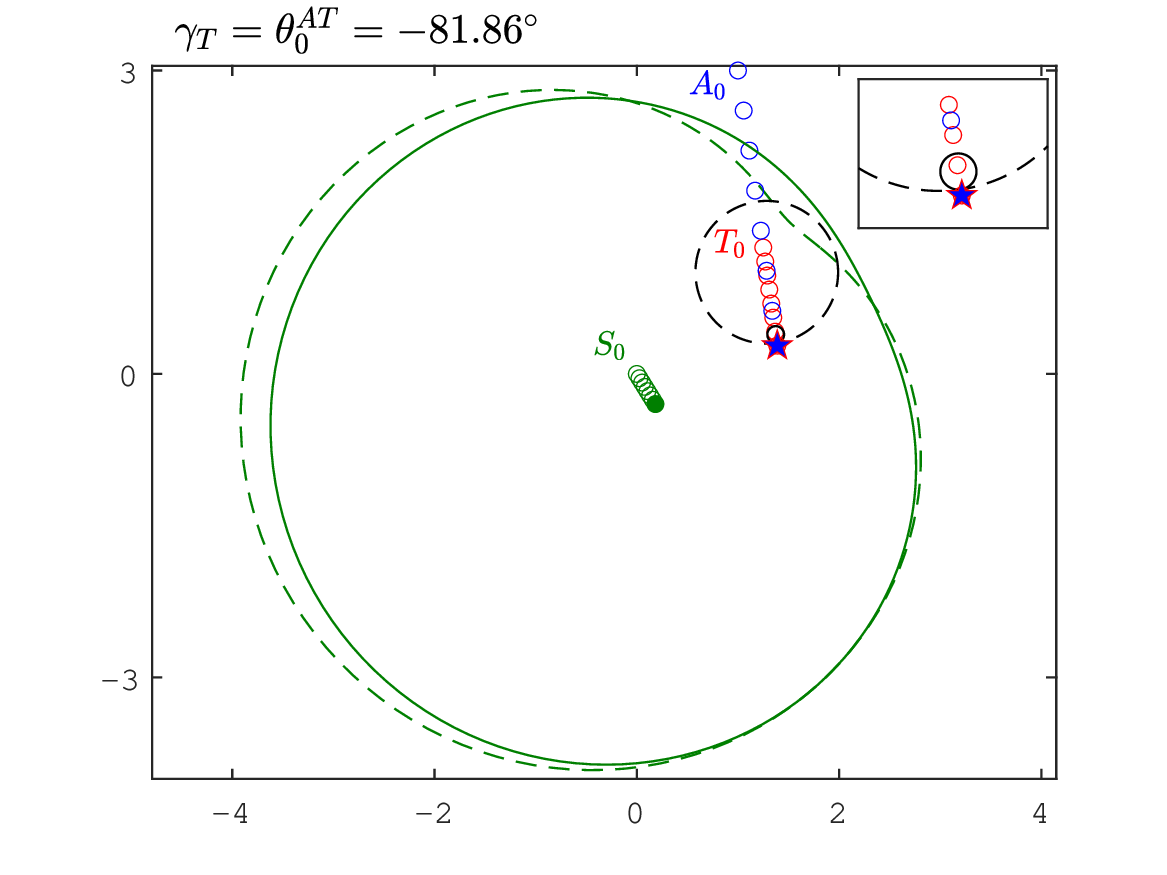}
        \caption{$\gamma^T = {\theta^{AT}_0} = -81.86^\circ$}
        \label{fig:gtAt}
    \end{subfigure}
    \begin{subfigure}[t]{0.325\textwidth}
        \centering
\includegraphics[width=\linewidth]{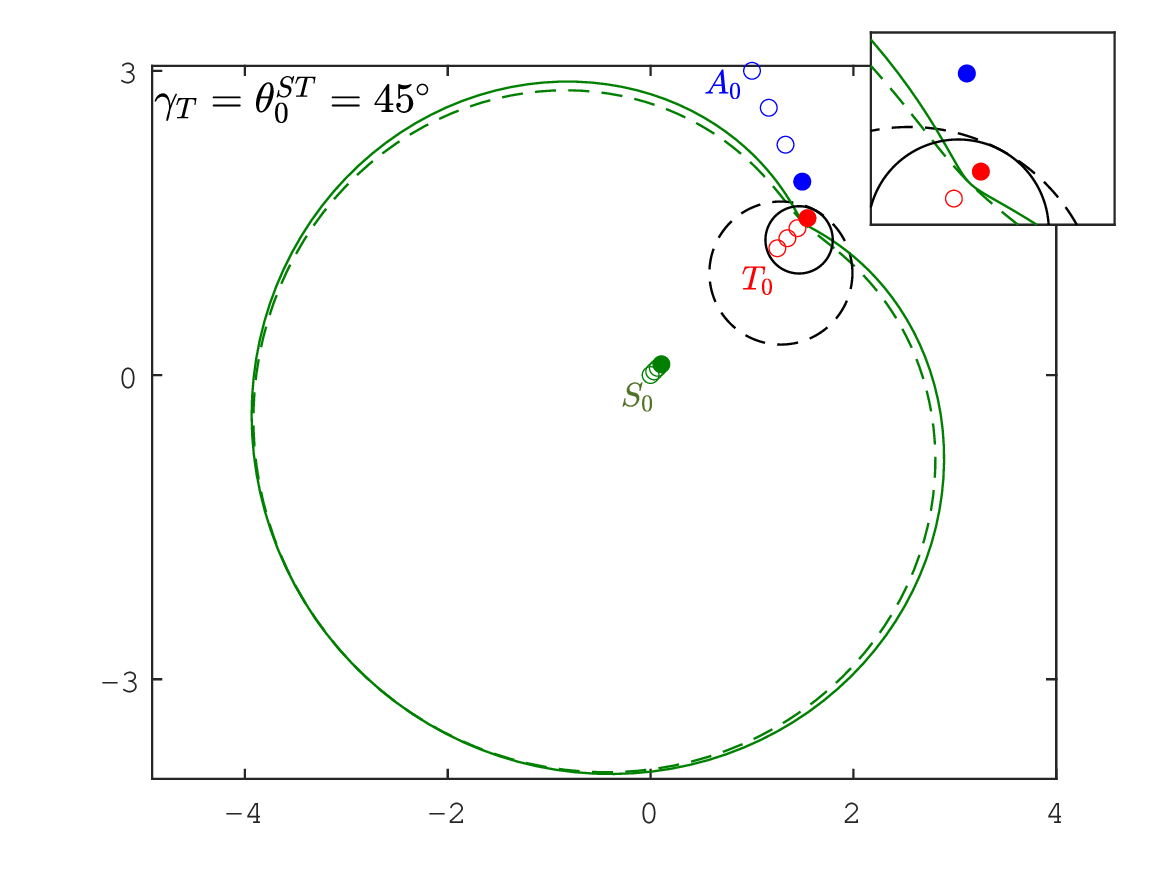}
        \caption{$\gamma^T ={\theta^{ST}_0} = 45^\circ$}
         \label{fig:gtSt}
    \end{subfigure}
    \begin{subfigure}[t]{0.325\textwidth}
        \centering
\includegraphics[width=\linewidth]{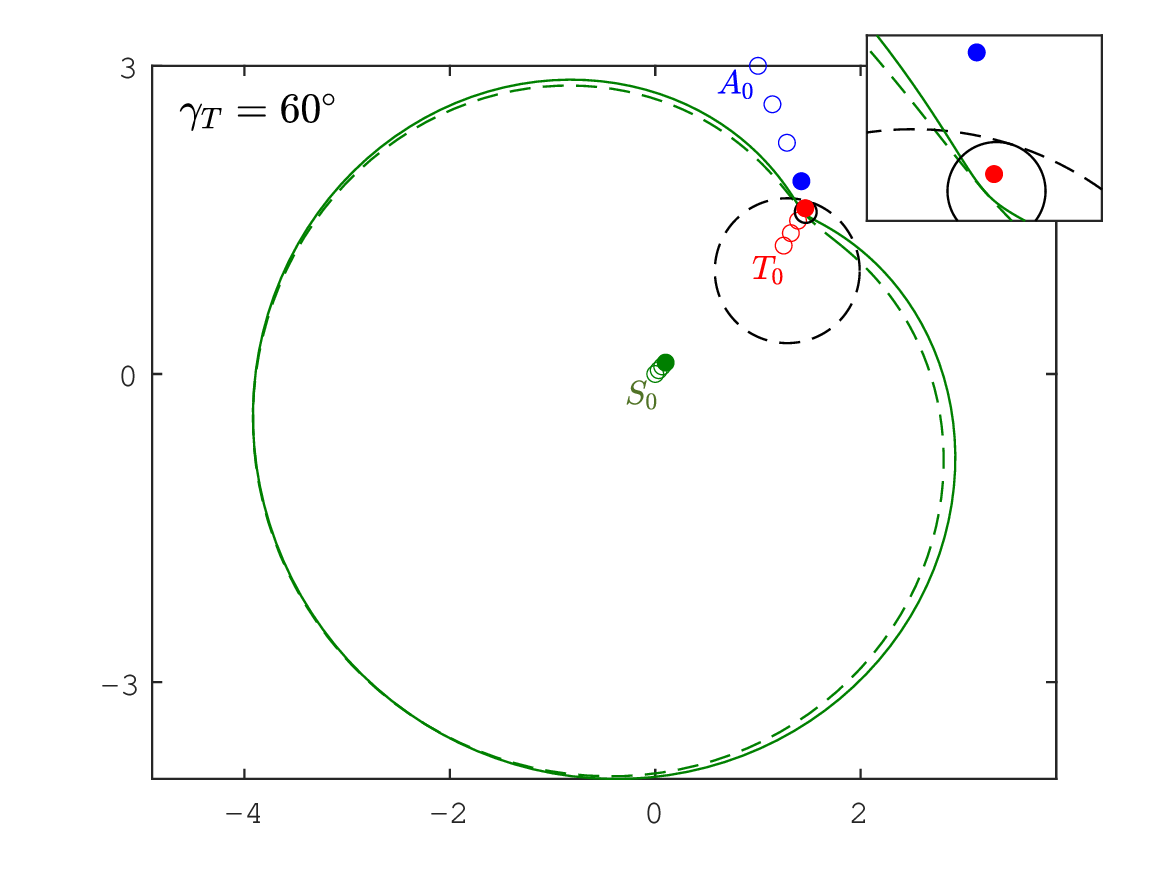}
        \caption{$\gamma^T = 60^\circ$}
         \label{fig:gt60}
    \end{subfigure}s
    \caption{Evolution of agent trajectories for various values of the target's heading angle $\gamma^T$. Each subplot illustrates the motion of the attacker ($A$), sensor ($S$), and target ($T$) over time, along with the set $\mathcal{A}_t$ the set $\mathcal{S}_t$. The sets $\mathcal{A}_t$ and $\mathcal{S}_t$ are shown at the initial time $t = 0$ (dashed line) and at the termination time $t = t_f$ (solid line).}
\label{fig:gt_variation}
\end{figure*}
Table~\ref{tab:thm1} presents the escape distances $\|T_0 T_{t_f}\|$ and the corresponding termination times $t_f$ for different values of the target's heading angle $\gamma^T$. Based on the trajectories, it is evident that under the given initial conditions, the Apollonius circle $\mathcal{A}_0$ intersects the sensor's sensable region $\mathcal{S}_0$, indicating that the target $T$ has viable escape strategies.

Figure \ref{fig:gtAt} shows a scenario when the target moves directly away from the attacker, that is, $\gamma^T = \theta^{AT}_0$. In this case, $T$ needs to cover a greater distance to escape, and is captured by $A$ due to its higher speed. In contrast, when the target moves directly away from the sensor, that is, $\gamma^T = \theta^{ST}_0 = 45^\circ$, it follows the shortest path and escapes successfully as shown in Fig. \ref{fig:gtSt}. It can also be confirmed from Table \ref{tab:thm1} that the time required to escape is minimum in this case. Additionally, it can be observed from Fig.~\ref{fig:gt60} that the target also manages to escape when $\gamma^T = 60^\circ$, although the time taken in this case is longer than when $\gamma^T = 45^\circ$.

\subsection{Visualization of Theorem \ref{thmVT}}
To visualize Theorem~\ref{thmVT}  we illustrate how the changes in  target speed $v^T$  relative to a critical threshold affects the escape conditions for $T$.

In the first case, the target speed is lower than the threshold given by Eq. (\ref{vTmax}), which leads to $\mathcal{A}_0 \subset \mathcal{S}_0$, guaranteeing that the target will eventually be captured. In the second case, the target moves with a speed greater than the critical value, resulting in the  $\mathcal{A}_0$ partially extending beyond the sensable region $\mathcal{S}_0$, thereby allowing an escape strategy for $T$ to exist.

The initial positions of the agents are: $A_0 = (-2,\, 1)$, $S_0 = ( 0,\,0 )$, and $T_0 = (1,\,0.5 )$. 
The attacker moves with speed $v^A = 1\,$m/s and the sensor is moving with $v^S = 0.125\,$m/s. The sensing radius is assumed to be limited to $R = 2\,$m. Additionally, the initial relative separations and line-of-sight angles  are given by $d^{AT}_0 =3.0414\,$m, $\theta^{AT}_0 = -9.4623^\circ$, and $d^{ST}_0 = 1.1180\,$m, $\theta^{ST}_0 = 26.56^\circ$. With these initial conditions, when Theorem \ref{thmVT} is applied it can be calculated that in order to have $\mathcal{A}_0 \subset \mathcal{S}_0$, the target speed $v^T$ should be less than $0.3217\,$m/s.

\begin{figure*}[!h]
    \centering
    \begin{subfigure}[t]{0.325\textwidth}
        \centering
        \includegraphics[width=\linewidth]{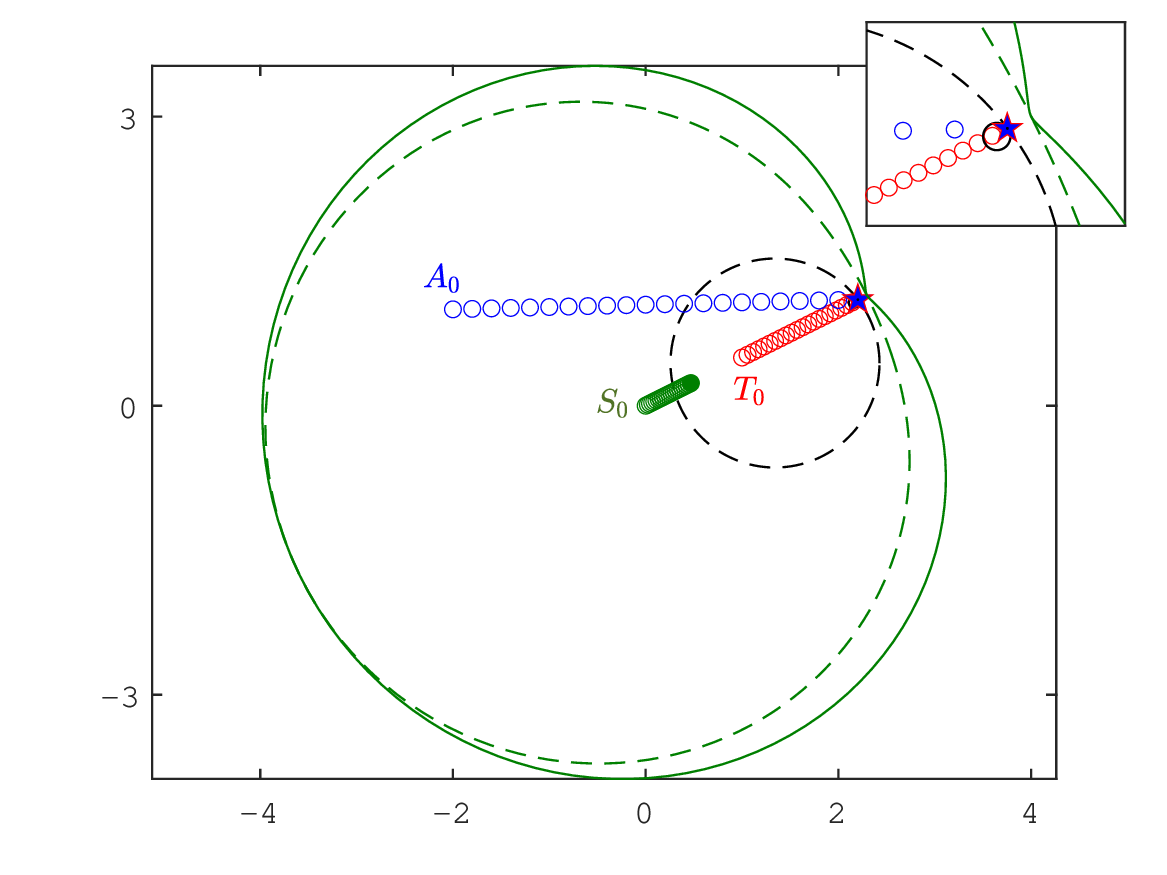}
        \caption{}
        \label{fig:int}
    \end{subfigure}
    \begin{subfigure}[t]{0.325\textwidth}
        \centering
    \includegraphics[width=\linewidth]{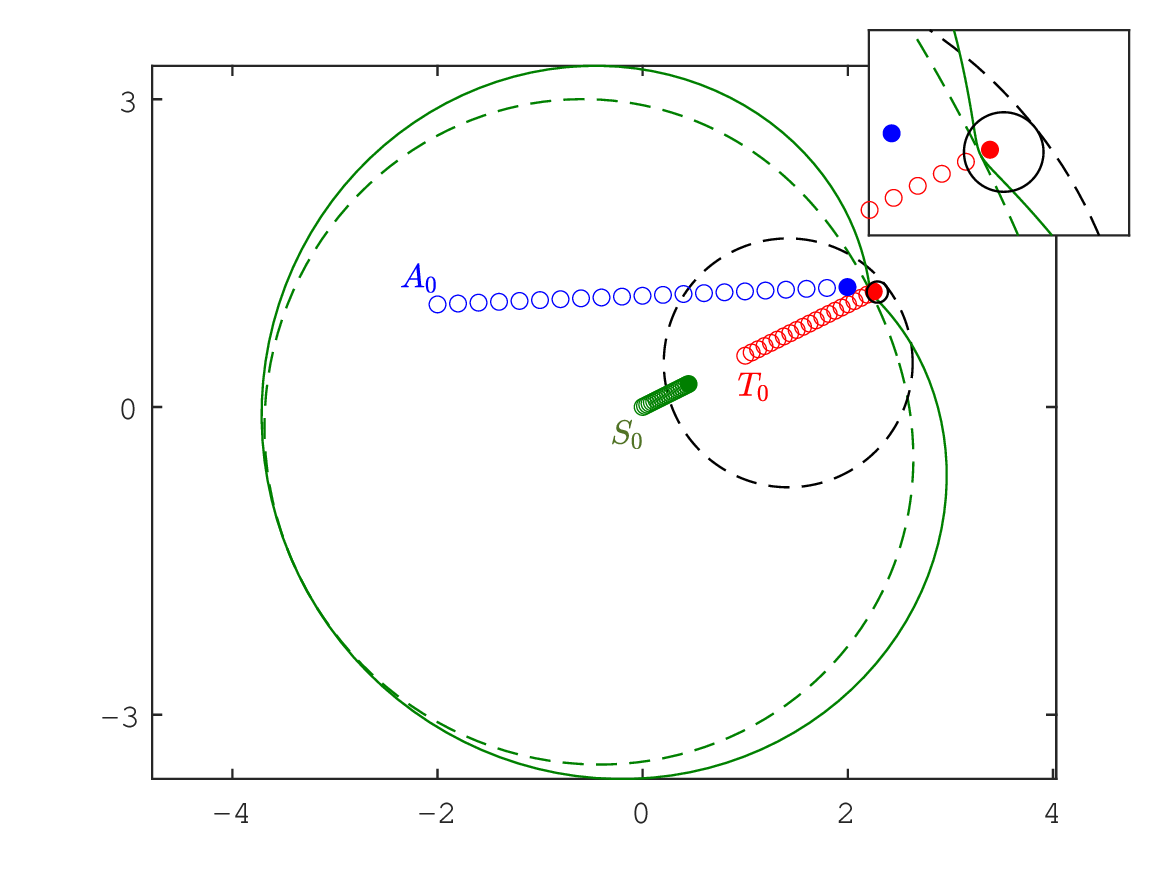}
        \caption{}
        \label{fig:esc}
    \end{subfigure}
    \begin{subfigure}[t]{0.325\textwidth}
        \centering
    \includegraphics[width=\linewidth]{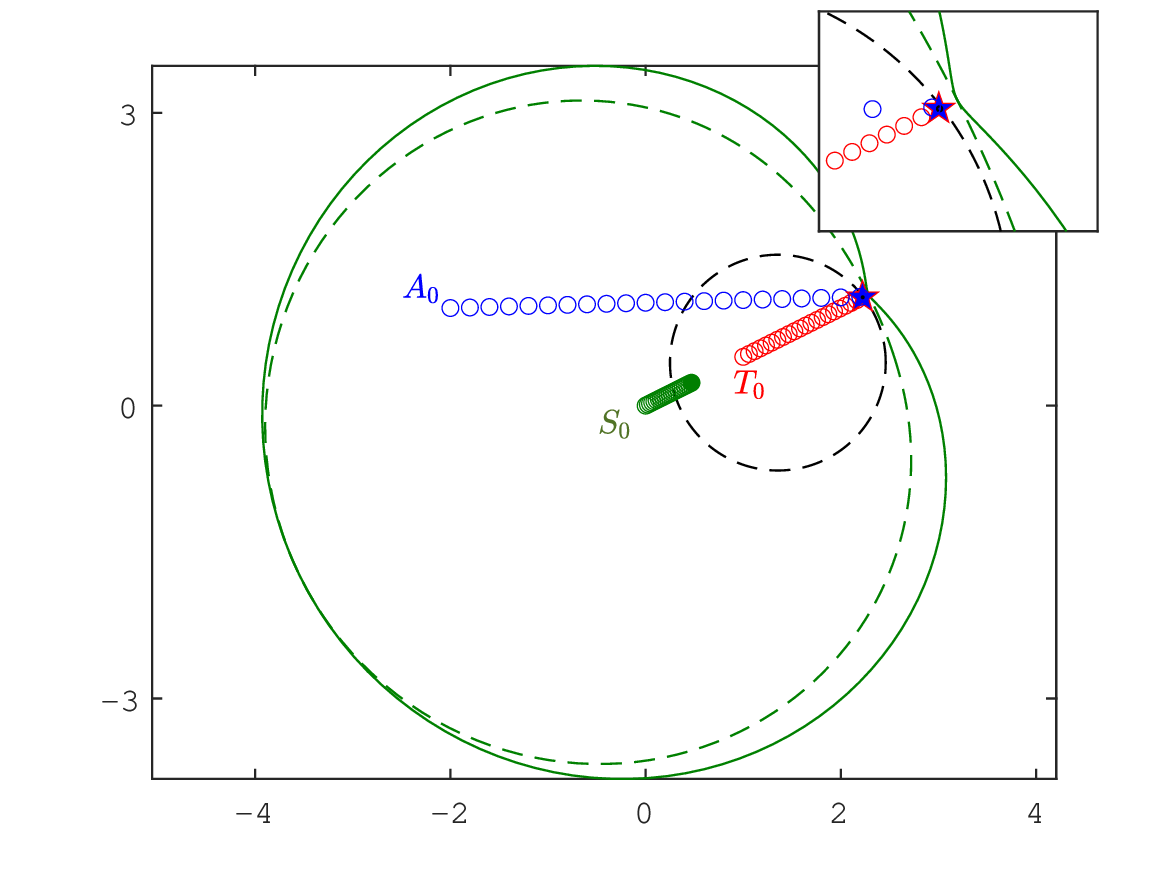}
        \caption{}
        \label{fig:ex}
    \end{subfigure}
    \caption{Illustration of Theorem~\ref{thmVT}. In subplot (a), the target's speed is below the threshold, resulting in $\mathcal{A}_0 \subset \mathcal{S}_0$ and leading to guaranteed capture. In subplot (b), the target's speed exceeds the critical threshold, causing the set $\mathcal{A}$ to extend beyond the sensable region $\mathcal{S}$, enabling the target to escape. (c) although the target speed exceeds the critical value, the condition $\mathcal{A}_0 \subset \mathcal{S}_0$ still holds, and capture remains guaranteed. In all three subplots the sets $\mathcal{A}_t$ and $\mathcal{S}_t$ are shown at the initial time $t = 0$ (dashed line) and at the termination time $t = t_f$ (solid line). }
    \label{fig:evthm2}
\end{figure*}

Figure~\ref{fig:evthm2} illustrates the outcomes of these two cases. In subplot Fig. \ref{fig:int}, the target moves directly away from the sensor (that is, $\gamma^T = \theta^{ST}_0$), which should allow it to escape in minimum time. However, the target speed, for this case, is taken as $v^T = 0.32\,$m/s. The target's speed is lower than given in Eq. (\ref{vTmax}), as a result, $\mathcal{A}_0$ lies entirely within $\mathcal{S}_0$. Therefore, by Theorem~\ref{thmVT}, this containment persists for all future times, guaranteeing the target’s capture -- even when it adopts the optimal strategy of moving directly away from the sensor.

In contrast, subplot Fig.~\ref{fig:esc} shows the scenario with $v^T = 0.35\,$m/s. At this higher speed, $\mathcal{A}_0$ extends beyond the boundary of $\mathcal{S}_0$, allowing the target to formulate a viable escape strategy. With all other parameters being the same, the simulation confirms that the engagement results in the target’s successful escape.

Using the same initial conditions, we conducted an additional simulation with the target speed set to $v^T = 0.3250\,$m/s. Although this speed exceeds the critical threshold, Fig.~\ref{fig:ex} demonstrates that $\mathcal{A}_0 \subset \mathcal{S}_0$ still holds, thereby enabling the attacker to successfully capture the target. This example illustrates that the condition presented in Theorem~\ref{thmVT} is sufficient to guarantee capture; however, the converse is not necessarily true --failure to satisfy the condition does not imply that the target is guaranteed to escape.

\subsection{Visualization of Theorem \ref{thmesc}}
The initial positions of the agents for this illustration are as follows: $S_0 = (0,\,0)$, $A_0 = (3,\,3)$, and $T_0 = (1.5,\, 0.5)$. The speeds of the sensor and the attacker are set to $v^S = 0.3\,$m/s and $v^A = 1\,$m/s, respectively. The initial line-of-sight angles and relative distances at $t = 0$ are given by $\theta^{ST}_0 = 18.43^\circ$, $d_0^{ST} = 1.5811\,$m, $\theta_0^{AT} = -120.96^\circ$, and $d_0^{AT} = 2.9155\,$m. The relative sensing range is the same as in the previous scenario. The sensor and the target follow their optimal strategies given in Eqs. (\ref{eq:gammastar}) and (\ref{gammaA}), respectively. The target heading is considered to be equal to $\gamma^T = \theta_0^{ST}$. 

Given these engagement parameters, application of Theorem~\ref{thmesc} indicates that if $v^T > 0.5181\,$m/s, then the Apollonius set $\mathcal{A}_0$ extends beyond the sensable region $\mathcal{S}_0$, making escape feasible for the target. We consider a target speed of $v^T = 0.5190\,$m/s. As illustrated in Fig.~\ref{fig:thm4esc}, it is evident that the Apollonius circle exceeds the sensable region, enabling the target to successfully escape.
\begin{figure}
    \centering
    \includegraphics[width=\linewidth]{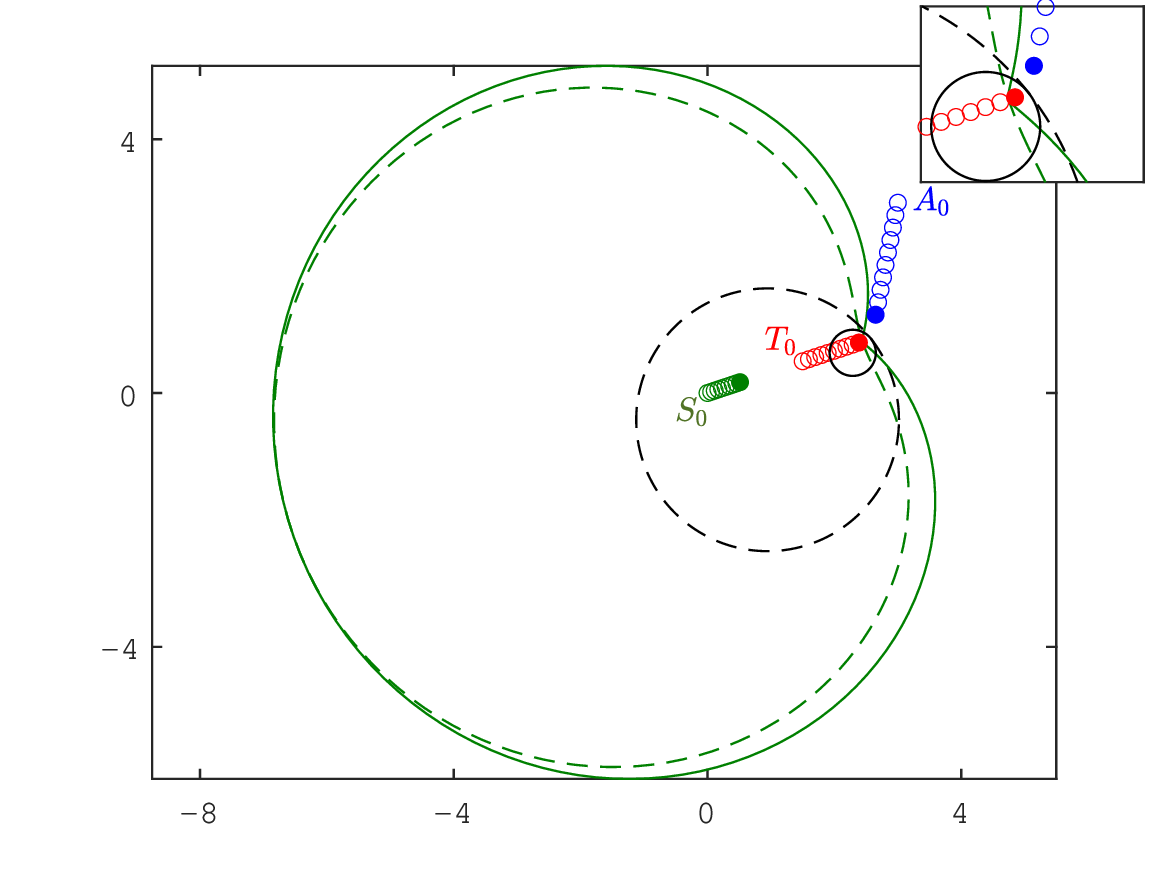}
    \caption{Illustration of Theorem \ref{thmesc}. The sets $\mathcal{A}_t$ and $\mathcal{S}_t$ are shown at the initial time $t = 0$ (dashed line) and at the termination time $t = t_f$ (solid line). Target escape for $v^T = 0.5190\,$m/s, since $\mathcal{A}_0$ extends beyond $\mathcal{S}_0$. }
    \label{fig:thm4esc}
\end{figure}

\subsection{Visualization of Theorem \ref{thm:vTstar_tangent}}
Theorem~\ref{thm:vTstar_tangent} provides a sharper bound on the target speed that determines the game's outcome. To illustrate this result, we consider the same engagement parameters and initial positions $S_0$, $A_0$, and $T_0$ as in the previous subsection. Substituting these values into the quadratic equation~(\ref{vTcrit}) yields two solutions: $\overline{v}^{T\star} = 0.1765\,$m/s and $\overline{v}^{T\star} = 0.4896\,$m/s. The admissible solution must satisfy $\underline{v}^T \leq \overline{v}^{T\star} \leq \overline{v}^T$. Using Eqs.~(\ref{vTmin}) and (\ref{vTmax}), we compute $\underline{v}^T = 0.3879\,$m/s and $\overline{v}^T = 0.5181\,$m/s. Therefore, the valid critical speed is $\overline{v}^{T\star} = 0.4896\,$m/s.

\begin{figure}
    \centering
    \includegraphics[width=\linewidth]{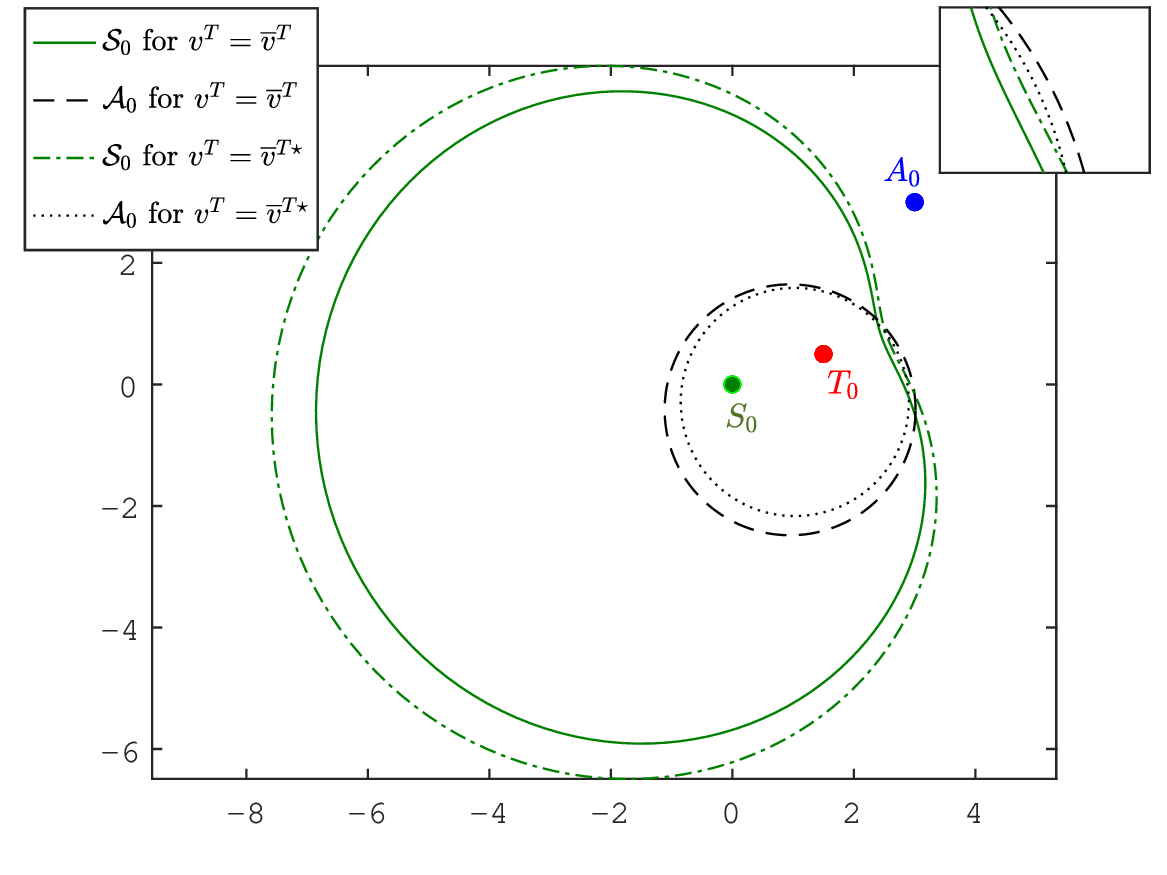}
   \caption{Illustration of Theorem~\ref{thm:vTstar_tangent} for $v^T = \overline{v}^T$ and $v^T = \overline{v}^{T\star}$. The figure shows the Apollonius circles $\mathcal{A}_0$ and the sensable regions $\mathcal{S}_0$ corresponding to both cases. The speed $\overline{v}^{T\star}$ provides a tighter bound on the critical escape speed by incorporating the geometric alignment of the agents.}
    \label{fig:thm6}
\end{figure}

Figure~\ref{fig:thm6} shows the behaviors of $\mathcal{A}_0$ and $\mathcal{S}_0$ for $v^T = \overline{v}^T$ and $v^T = \overline{v}^{T\star}$. In both cases, $\mathcal{A}_0$ extends beyond $\mathcal{S}_0$, indicating that the target has an escape strategy. However, the bound with $v^T = \overline{v}^{T\star}$ is less conservative compared to $\overline{v}^T$. Thus, $\overline{v}^{T\star}$ offers a tighter estimate of the critical speed $v^{T\star}$, discussed in Theorem \ref{thm:vTcrit}, that guarantees either escape or capture. Overall, the result confirms that incorporating the orientation between the agents allows for a sharper estimate of the minimum speed required for the target to escape.

 \section{Conclusion}\label{SecVI}
We introduced a novel problem of cooperative pursuit between a heterogeneous pair of agents such as a sensor-attacker team and a mobile target. The sensor seeks to keep the target within a sensing distance from itself, while the attacker seeks to complete the capture but only as long as the target is within sensing distance of the sensor. The target moves to make its distance from the sensor to exceed the sensing radius. We presented an optimal strategy for the sensor to maximize the time for which it can keep the target within its sensing radius, while the attacker uses an existing proportional navigation strategy for capture. We derived sufficient conditions under which this proposed set of strategies leads to capture. Specifically, we derived explicit conditions on the target speed as a function of the problem parameters which guarantees capture and which guarantees evasion from any initial player configurations. Furthermore, we established the existence of a critical target speed that serves as a necessary and sufficient condition, precisely determining the outcome of capture or evasion. We also characterized a tighter sufficient condition on the target speed for escape by considering the specific initial orientation between the agents.

Our future work will include tightening the analysis to derive necessary and sufficient condition on the target speed for capture and for evasion. Other variations of this theme of heterogeneous pursuit is a topic of future work. {Moreover, the strategies developed in this work are limited to 2D settings. Extending them to 3D introduces significant geometric complexity, especially in modeling visibility and motion, and will be explored in future work.}

\section*{Acknowledgments}
We thank Satyanarayana Gupta Manyam from DCS Corp. ~for his comments and conversations on this work.

\ifCLASSOPTIONcaptionsoff
  \newpage
\fi

 \bibliographystyle{ieeetr}
 \bibliography{ref}

\begin{IEEEbiography}[{\includegraphics[width=1in,height=1.25in,clip,keepaspectratio]{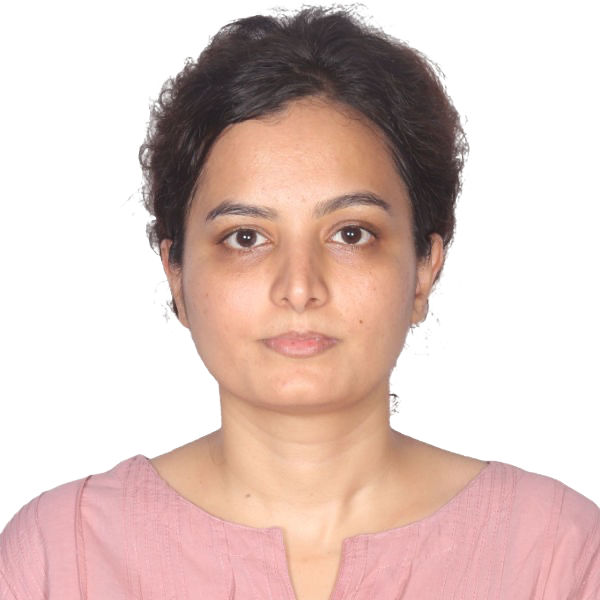}}]{Prajakta Surve}{\space}
~is a Post-Doctoral Researcher with the Electrical and Computer Engineering Department at Michigan State University, USA. Her research interests include guidance and control of aerospace vehicles, pursuit-evasion problems,  motion planning under sensing constraints and  game theory. She received the Bachelor of Engineering (B.E.) degree in Mechanical Engineering from Pimpri-Chinchwad College of Engineering, Pune, India, in 2014, and the Master of Technology (M.Tech.) degree in Aerospace Engineering with a specialization in Unmanned Aerial Vehicles from the University of Petroleum and Energy Studies, Dehradun, India, in 2017.
\end{IEEEbiography}

\begin{IEEEbiography}[{\includegraphics[width=1in,height=1.25in,clip,keepaspectratio]{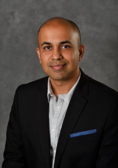}}]{Shaunak D. Bopardikar}
	~(S'08 -- M'10 -- SM'19) is an Associate Professor with the Electrical and Computer Engineering Department at the Michigan State University, USA. His research interests include autonomous motion planning and control, cyber-physical security, and scalable optimization and computation. He received the Bachelor of Technology (B.Tech.) and Master of Technology (M.Tech.) degrees in Mechanical Engineering from the Indian Institute of Technology, Bombay, India, in 2004, and the Ph.D. degree in Mechanical Engineering from the University of California at Santa Barbara, USA, in 2010. Dr. Bopardikar has over 85 refereed journal and conference publications, and holds 2 U.S.~patented inventions. His recognitions include an Air Force Research Laboratory Summer Faculty Fellowship, an IEEE Technical Committee on Security and Privacy's best student paper award (as advisor), a National Science Foundation Career Award and a Michigan State University College of Engineering's Withrow Teaching Excellence Award.
\end{IEEEbiography}

\begin{IEEEbiography}[{\includegraphics[width=1in,height=1.25in,clip,keepaspectratio]{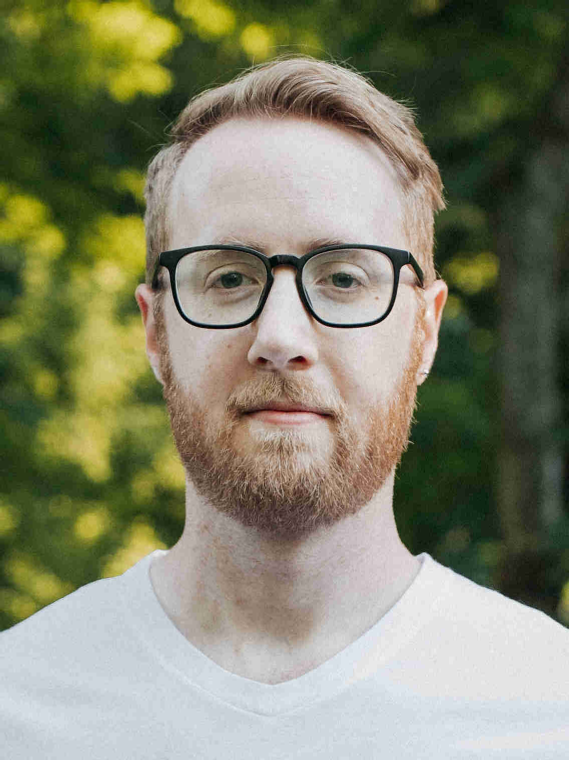}}]{Alexander Von Moll}
    ~is a researcher with the Control Science Center, Aerospace Systems Directorate, Air Force Research Laboratory. He holds a B.S. in Aerospace Engineering from Ohio State (2012), an M.S. in Aerospace Engineering from Georgia Institute of Technology (2016), and a Ph.D. in Electrical Engineering from University of Cincinnati (2022). Alex was a Department of Defense SMART Scholar, awarded in 2011 and again in 2014. His research interests include multi-agent systems, cooperative control, and differential games.
\end{IEEEbiography}
\begin{IEEEbiography}
[{\includegraphics[width=1in,height=1.25in,clip,keepaspectratio]{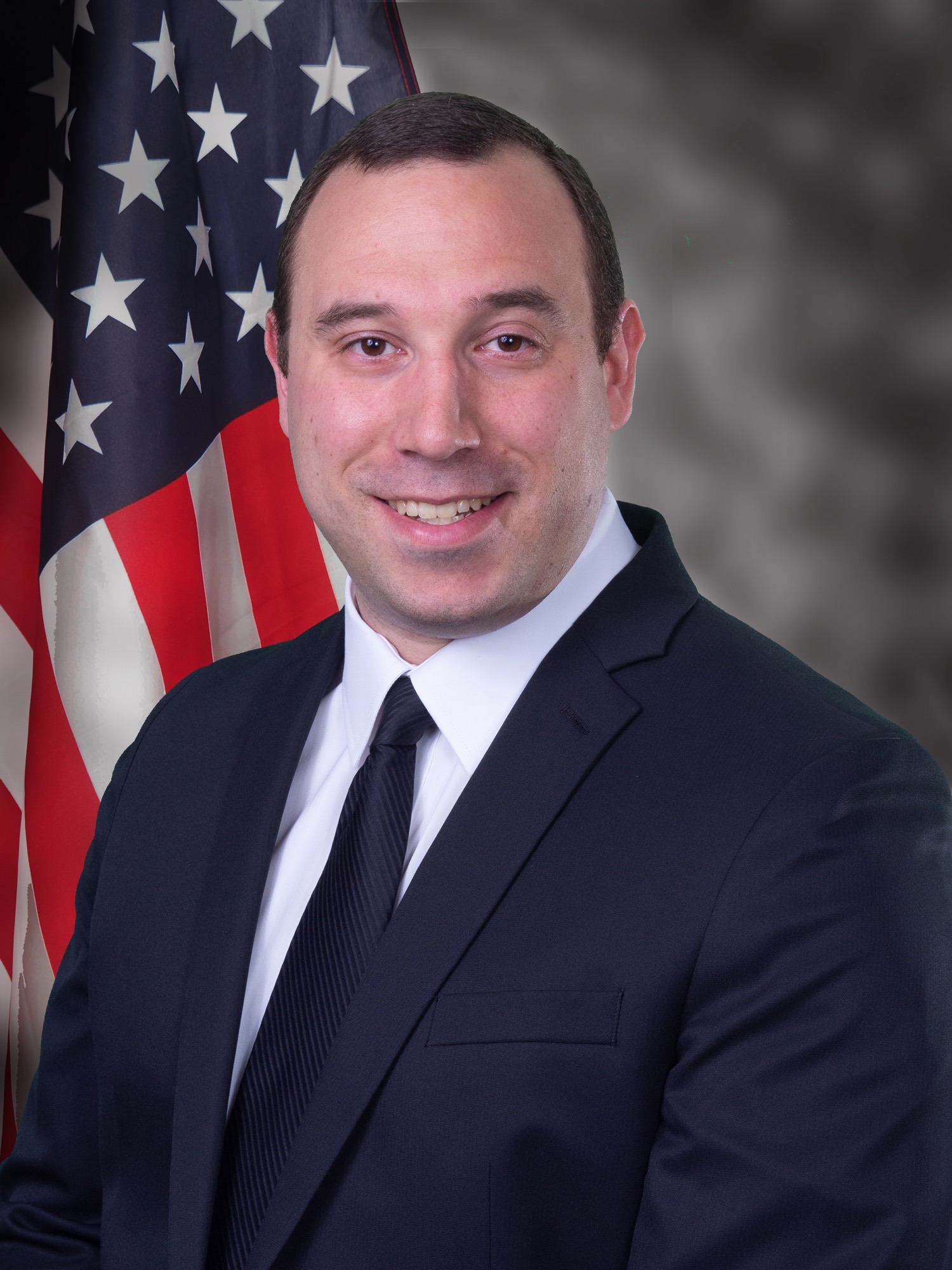}}]{Isaac Weintraub}
    ~is a Senior Electronics Engineer in the Autonomous Controls branch in the Power and Control division of the Aerospace Systems Directorate in the Air Force Research Laboratory. He received a PhD in Electrical Engineering from the Air Force Institute of Technology in 2021, an MS in Electrical Engineering from the University of Texas at Arlington in 2011, and a BS in Mechanical Engineering from Rose-Hulman Institute of Technology in 2009. He is currently an Associate Fellow of the American Institute of Aeronautics and Astronautics (AIAA) and a senior member of the Institute of Electrical and Electronics Engineers (IEEE). He is a member of the Intelligent Systems Technical Committee (ISTC) in the AIAA and an active member of both the Robotics and Automation Society (RAS) and the Controls Systems Society (CSS) in the IEEE. His research interests lie in automation and control of aerospace systems for defense applications.
\end{IEEEbiography}

\begin{IEEEbiography}[{\includegraphics[width=1in,height=1.25in,clip,keepaspectratio]{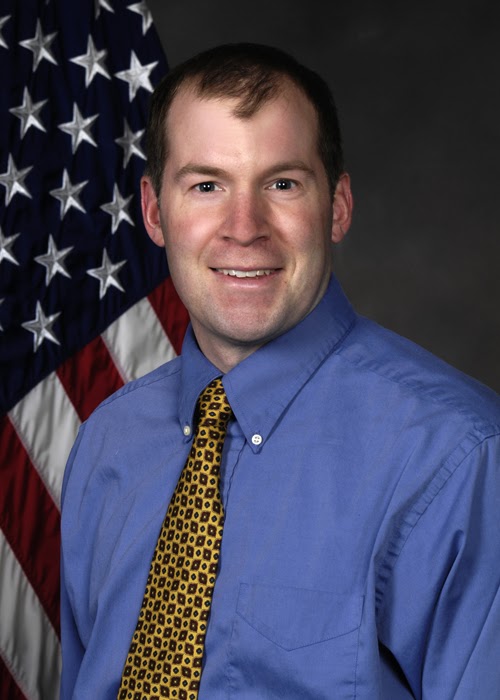}}]{David W. Casbeer}
      ~ is the Technical Area Lead over Cooperative \& Intelligent UAV Control with the Control Science Center, Aerospace Systems Directorate, Air Force Research Laboratory, where he carries out and leads basic research involving the control of autonomous UAVs with a particular emphasis on high-level decision making and planning under uncertainty. He received B.S. and Ph.D. degrees in Electrical Engineering from Brigham Young University in 2003 and 2009, respectively. 
\end{IEEEbiography}

\end{document}